\documentclass{article}
\pdfoutput=1
\PassOptionsToPackage{numbers,sort,compress}{natbib}

\usepackage[preprint]{nips_2018}

\usepackage[utf8]{inputenc} 
\usepackage[T1]{fontenc}    
\usepackage{hyperref}       
\usepackage{url}            
\usepackage{booktabs}       
\usepackage{amsfonts}       
\usepackage{nicefrac}       
\usepackage{microtype}      

\usepackage{microtype}
\usepackage{booktabs} 
\usepackage{epsf}
\usepackage{epsfig}
\usepackage{subfigure}
\usepackage{amsmath}
\usepackage{amssymb}
\usepackage{amsthm}
\usepackage{mathtools}
\usepackage{multirow}
\usepackage{multicol}
\usepackage{graphicx} 
\usepackage{xspace}
\usepackage{listings}
\usepackage{xcolor}

\usepackage{enumitem}

\title{Phocas: dimensional Byzantine-resilient stochastic gradient descent}

\author{
  Cong Xie \\
  Department of Computer Science\\
  University of Illinois at Urbana Champaign\\
  \texttt{cx2@illinois.edu} \\
  \And
  Oluwasanmi Koyejo \\
  Department of Computer Science\\
  University of Illinois at Urbana Champaign\\
  \texttt{cx2@illinois.edu} \\
  \And
  Indranil Gupta \\
  Department of Computer Science\\
  University of Illinois at Urbana Champaign\\
  \texttt{cx2@illinois.edu} \\
}

\begin{document}

\def\Blue{\color{blue}}
\def\Purple{\color{purple}}

\def\A{{\bf A}}
\def\a{{\bf a}}
\def\B{{\bf B}}
\def\C{{\bf C}}
\def\c{{\bf c}}
\def\D{{\bf D}}
\def\d{{\bf d}}
\def\F{{\bf F}}
\def\e{{\bf e}}
\def\f{{\bf f}}
\def\G{{\bf G}}
\def\H{{\bf H}}
\def\I{{\bf I}}
\def\K{{\bf K}}
\def\L{{\bf L}}
\def\M{{\bf M}}
\def\m{{\bf m}}
\def\N{{\bf N}}
\def\n{{\bf n}}
\def\Q{{\bf Q}}
\def\q{{\bf q}}
\def\S{{\bf S}}
\def\s{{\bf s}}
\def\T{{\bf T}}
\def\U{{\bf U}}
\def\u{{\bf u}}
\def\V{{\bf V}}
\def\v{{\bf v}}
\def\W{{\bf W}}
\def\w{{\bf w}}
\def\X{{\bf X}}
\def\x{{\bf x}}
\def\Y{{\bf Y}}
\def\y{{\bf y}}
\def\Z{{\bf Z}}
\def\z{{\bf z}}
\def\0{{\bf 0}}
\def\1{{\bf 1}}

\def\AM{{\mathcal A}}
\def\CM{{\mathcal C}}
\def\DM{{\mathcal D}}
\def\GM{{\mathcal G}}
\def\FM{{\mathcal F}}
\def\IM{{\mathcal I}}
\def\NM{{\mathcal N}}
\def\OM{{\mathcal O}}
\def\SM{{\mathcal S}}
\def\TM{{\mathcal T}}
\def\UM{{\mathcal U}}
\def\XM{{\mathcal X}}
\def\YM{{\mathcal Y}}
\def\RB{{\mathbb R}}

\def\TX{\tilde{\bf X}}
\def\tx{\tilde{\bf x}}
\def\ty{\tilde{\bf y}}
\def\TZ{\tilde{\bf Z}}
\def\tz{\tilde{\bf z}}
\def\hd{\hat{d}}
\def\HD{\hat{\bf D}}
\def\hx{\hat{\bf x}}
\def\TD{\tilde{\Delta}}

\def\alp{\mbox{\boldmath$\alpha$\unboldmath}}
\def\bet{\mbox{\boldmath$\beta$\unboldmath}}
\def\epsi{\mbox{\boldmath$\epsilon$\unboldmath}}
\def\etab{\mbox{\boldmath$\eta$\unboldmath}}
\def\ph{\mbox{\boldmath$\phi$\unboldmath}}
\def\pii{\mbox{\boldmath$\pi$\unboldmath}}
\def\Ph{\mbox{\boldmath$\Phi$\unboldmath}}
\def\Ps{\mbox{\boldmath$\Psi$\unboldmath}}
\def\tha{\mbox{\boldmath$\theta$\unboldmath}}
\def\Tha{\mbox{\boldmath$\Theta$\unboldmath}}
\def\muu{\mbox{\boldmath$\mu$\unboldmath}}
\def\Si{\mbox{\boldmath$\Sigma$\unboldmath}}
\def\si{\mbox{\boldmath$\sigma$\unboldmath}}
\def\Gam{\mbox{\boldmath$\Gamma$\unboldmath}}
\def\Lam{\mbox{\boldmath$\Lambda$\unboldmath}}
\def\De{\mbox{\boldmath$\Delta$\unboldmath}}
\def\Ome{\mbox{\boldmath$\Omega$\unboldmath}}
\def\TOme{\mbox{\boldmath$\hat{\Omega}$\unboldmath}}
\def\vps{\mbox{\boldmath$\varepsilon$\unboldmath}}
\newcommand{\ti}[1]{\tilde{#1}}
\def\Ncal{\mathcal{N}}
\def\argmax{\mathop{\rm argmax}}
\def\argmin{\mathop{\rm argmin}}
\providecommand{\abs}[1]{\lvert#1\rvert}
\providecommand{\norm}[2]{\lVert#1\rVert_{#2}}

\def\Zs{{\Z_{\mathrm{S}}}}
\def\Zl{{\Z_{\mathrm{L}}}}
\def\Yr{{\Y_{\mathrm{R}}}}
\def\Yg{{\Y_{\mathrm{G}}}}
\def\Yb{{\Y_{\mathrm{B}}}}
\def\Ar{{\A_{\mathrm{R}}}}
\def\Ag{{\A_{\mathrm{G}}}}
\def\Ab{{\A_{\mathrm{B}}}}
\def\As{{\A_{\mathrm{S}}}}
\def\Asr{{\A_{\mathrm{S}_{\mathrm{R}}}}}
\def\Asg{{\A_{\mathrm{S}_{\mathrm{G}}}}}
\def\Asb{{\A_{\mathrm{S}_{\mathrm{B}}}}}
\def\Or{{\Ome_{\mathrm{R}}}}
\def\Og{{\Ome_{\mathrm{G}}}}
\def\Ob{{\Ome_{\mathrm{B}}}}

\def\vec{\mathrm{vec}}
\def\fold{\mathrm{fold}}
\def\index{\mathrm{index}}
\def\sgn{\mathrm{sgn}}
\def\tr{\mathrm{tr}}
\def\rk{\mathrm{rank}}
\def\diag{\mathsf{diag}}
\def\const{\mathrm{Const}}
\def\dg{\mathsf{dg}}
\def\st{\mathsf{s.t.}}
\def\vect{\mathsf{vec}}
\def\MCAR{\mathrm{MCAR}}
\def\MSAR{\mathrm{MSAR}}
\def\etal{{\em et al.\/}\,}
\newcommand{\indep}{{\;\bot\!\!\!\!\!\!\bot\;}}

\def\Lsize{\hbox{\space \raise-2mm\hbox{$\textstyle \L \atop \scriptstyle {m\times 3n}$} \space}}
\def\Ssize{\hbox{\space \raise-2mm\hbox{$\textstyle \S \atop \scriptstyle {m\times 3n}$} \space}}
\def\Osize{\hbox{\space \raise-2mm\hbox{$\textstyle \Ome \atop \scriptstyle {m\times 3n}$} \space}}
\def\Tsize{\hbox{\space \raise-2mm\hbox{$\textstyle \T \atop \scriptstyle {3n\times n}$} \space}}
\def\Bsize{\hbox{\space \raise-2mm\hbox{$\textstyle \B \atop \scriptstyle {m\times n}$} \space}}

\newcommand{\twopartdef}[4]
{
	\left\{
		\begin{array}{ll}
			#1 & \mbox{if } #2 \\
			#3 & \mbox{if } #4
		\end{array}
	\right.
}

\newcommand{\tabincell}[2]{\begin{tabular}{@{}#1@{}}#2\end{tabular}}

\newtheorem{theorem}{Theorem}
\newtheorem{lemma}{Lemma}
\newtheorem{proposition}{Proposition}
\newtheorem{corollary}{Corollary}
\newtheorem{definition}{Definition}
\newtheorem{remark}{Remark}

\def\E{{\mathbb E}}
\def\R{{\mathbb R}}

\DeclarePairedDelimiter\ceil{\lceil}{\rceil}
\DeclarePairedDelimiter\floor{\lfloor}{\rfloor}

\newcommand{\ip}[2]{\left\langle #1, #2 \right \rangle}

\newcommand{\comment}[1]{\text{\quad$\triangleright$ \textit{#1}}}

\def\aggr{{\tt Aggr}}
\def\krum{{\tt Krum}}
\def\mkrum{{\tt Multi-Krum}}
\def\mean{{\tt mean}}
\def\trmean{{\tt Trmean}}
\def\phocas{{\tt Phocas}}

\maketitle

\begin{abstract}
We propose a novel robust aggregation rule for distributed synchronous Stochastic Gradient Descent~(SGD) under a general Byzantine failure model. The attackers can arbitrarily manipulate the data transferred between the servers and the workers in the parameter server~(PS) architecture. We prove the Byzantine resilience of the proposed aggregation rules. Empirical analysis shows that the proposed techniques outperform current approaches for realistic use cases and Byzantine attack scenarios.
\end{abstract}

\section{Introduction}

The failure resilience of distributed machine-learning systems has attracted increasing attention~\citep{blanchard2017machine,chen2017distributed,yin2018byzantine,alistarh2018byzantine} in the community. 
Larger clusters can accelerate training. However, this makes the distributed system more vulnerable to different kinds of failures or even attacks~\citep{harinath2017review}. 
Thus, failure/attack resilience is becoming more and more important for distributed machine-learning systems, especially for large-scale deep learning~\citep{Dean2012LargeSD,McMahan2017CommunicationEfficientLO}.  

In this paper, we consider the most general failure model, Byzantine failures~\citep{Lamport1982TheBG}, where the attackers can know any information of the other processes, and attack any value in transmission. To be more specific,  the data transmission between the machines can be replaced by arbitrary values. Under such model, there are no constraints on the failures or attackers. 


The distributed training framework studied in this paper is the \underline{P}arameter \underline{S}erver~(PS). The PS architecture is composed of the server nodes and the worker nodes. The server nodes maintain a global copy of the model, aggregate the gradients from the workers, apply the gradients to the model, and broadcast the latest model to the workers. The worker nodes pull the latest model from the server nodes, compute the gradients according to the local portion of the training data, and send the gradients to the server nodes. The entire dataset and the corresponding workload is distributed to multiple worker nodes, thus parallelizing the computation via partitioning the dataset. 

In this paper, we study the Byzantine resilience of synchronous Stochastic Gradient Descent~(SGD), which is a popular class of learning algorithms using PS architecture. Its variants are widely used in training deep neural networks~\cite{Kingma2014AdamAM,Mukkamala2017VariantsOR}. Such algorithms always wait to collect gradients from all the worker nodes before moving on to the next iteration. 

\begin{figure}[htb!]
\centering
\subfigure[Classic Byzantine]{\includegraphics[width=0.27\textwidth, height=0.09\textheight]{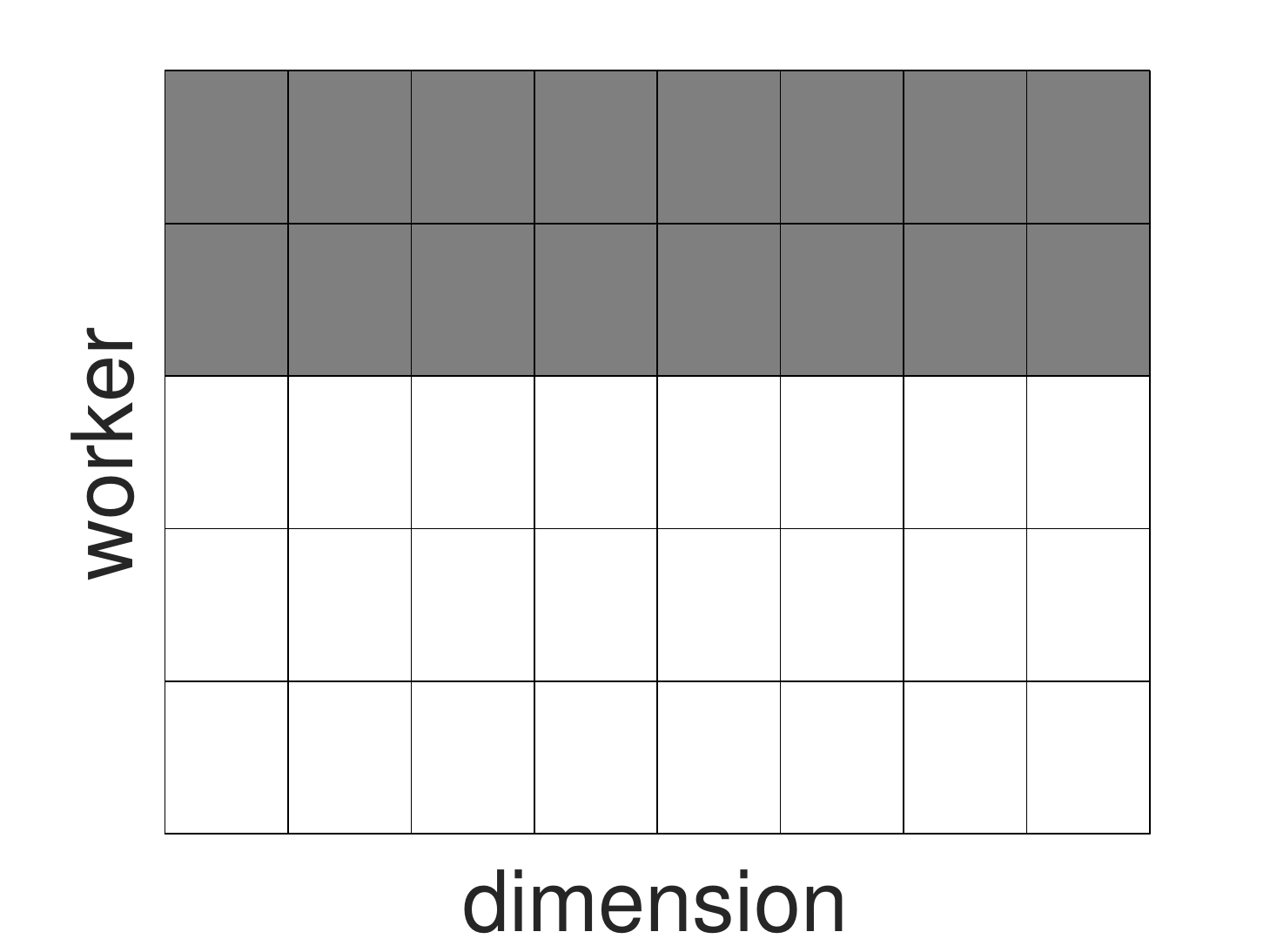}}
\subfigure[Generalized Byzantine]{\includegraphics[width=0.27\textwidth, height=0.09\textheight]{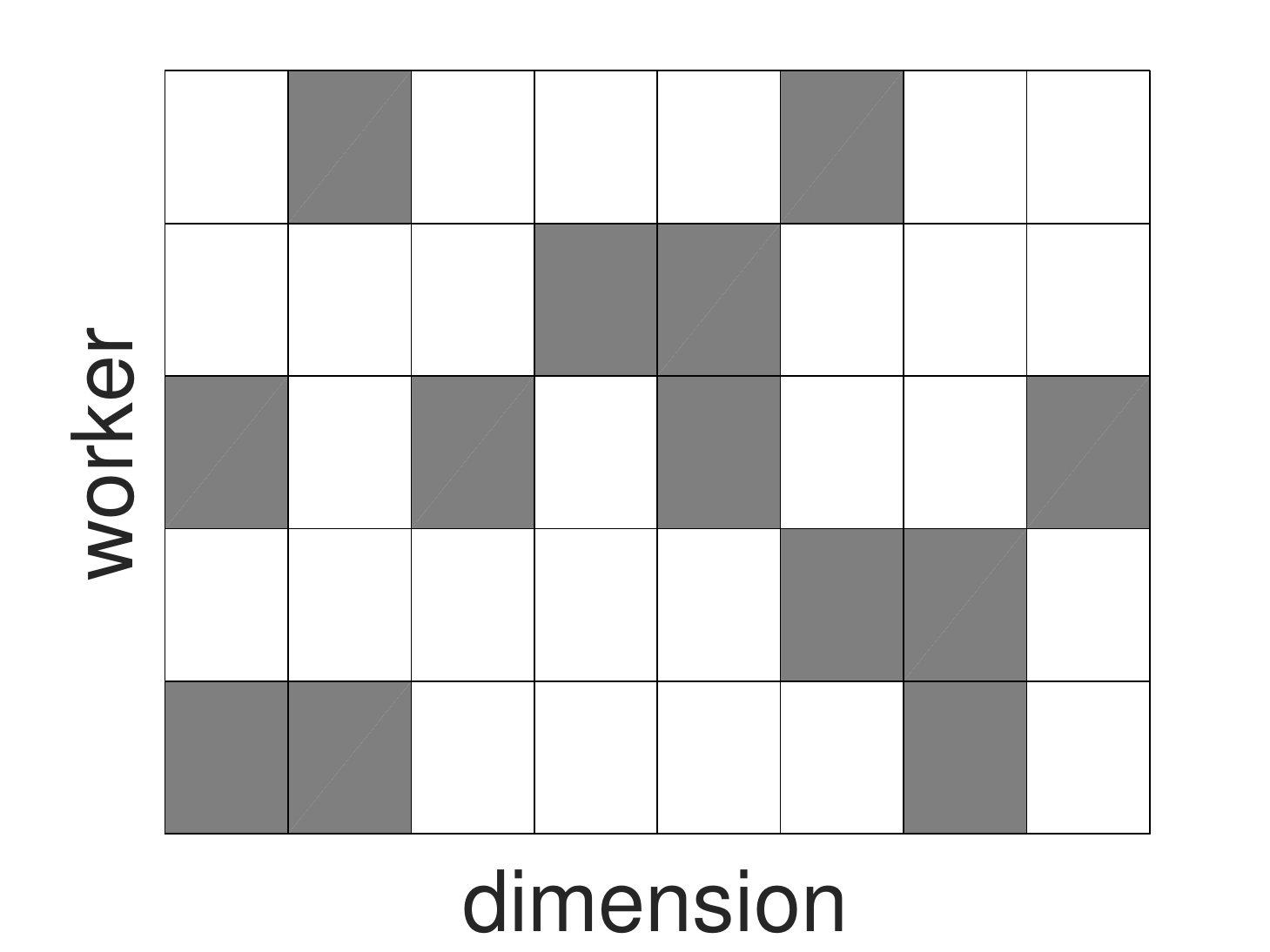}}
\caption{The 2 figures visualize $5$ workers with $8$-dimensional gradients. The $i$th row represents the gradient vector produced by the $i$th worker. The $j$th column represents the $j$th dimension of the gradients. A shadow block represents that the corresponding value is replaced by a Byzantine value. In the two examples, the maximal number of Byzantine values for each dimension is $2$. For the classic Byzantine model, all the Byzantine values must lie in the same workers~(rows), while for the generalized Byzantine model there is no such constraint. Thus, (a) is a special case of (b).}
\label{fig:viz_byz}
\end{figure}

The failure model can be described by using an $m \times d$ matrix consisting of the $d$-dimensional gradients produced by $m$ workers, as visualized in Figure~\ref{fig:viz_byz}. Previous work~\cite{blanchard2017machine} has so far addressed a special case, where the Byzantine values must lie in the same rows~(workers) as shown in Figure~\ref{fig:viz_byz}(a). Our failure model generalizes the classic Byzantine failure model by placing the Byzantine values anywhere in the matrix without any constraint. For example, \citet{lee2017risk} describes the vulnerability to bit-flipping attacks of a wireless transmission technology. The servers could receive data via such vulnerable communication media, even if the messages are encrypted. As a result, an arbitrary fraction of the received values are Byzantine.

There are two limitations lying in most of the existing Byzantine-resilient SGD algorithms~\cite{blanchard2017machine,chen2017distributed}. First, they only consider the classic Byzantine model shown in Figure~\ref{fig:viz_byz}(a). However, the Byzantine failures can also happen in the communication media/interfaces on the server side, which yields the generalized Byzantine model shown in Figure~\ref{fig:viz_byz}(b). Second, the algorithms are based on the Euclidean norm, which suffers from the curse of dimensionality. When the dimension gets higher, it will be more difficult to distinguish the Byzantine gradients from the correct ones.

In this paper, we study the dimensional Byzantine-resilient algorithms, which tolerate the generalized Byzantine model under certain conditions, and are not affected by the curse of dimensionality. 
We propose Byzantine-resilient trimmed-mean-based aggregation rules. We assume that for each dimension, the number of Byzantine values must be less than the number of correct ones. The resilience to such Byzantine model is called ``dimensional Byzantine resilience".
The main contributions of this paper are listed below: 
\setitemize[0]{leftmargin=*}
\begin{itemize} 
\item We formulate the dimensional Byzantine resilience property, and prove that the proposed trimmed-mean-based approaches are dimensional Byzantine-resilient~(Definition~\ref{def:dim_byz}). As far as we know, this paper is the first one to study generalized Byzantine failures and dimensional Byzantine resilience for synchronous SGD.
\item We show that the proposed aggregation rules have low computation cost. The time complexities are nearly linear, which are of the same order as averaging--the default choice for non-Byzantine aggregation.
\end{itemize}

\section{Model}

We consider the following optimization problem:
\begin{align*}
\min_{x} F(x), 
\end{align*}
where $F(x) = \E_{z \sim \mathcal{D}}[f(x; z)]$, $z$ is sampled from some unknown distribution $\mathcal{D}$. We assume that there exists a minimizer of $F(x)$, which is denoted by $x^*$.

We solve this problem in a distributed manner with $m$ workers. In each iteration, each worker will sample $n$ independent and identically distributed~(i.i.d.) data points from the distribution $\mathcal{D}$, and compute the gradient of the local empirical loss $F_i(x) = \frac{1}{n} \sum_{j=1}^n f(x; z^{i,j}), \forall i \in [m]$, where $z^{i,j}$ is the $j$th sampled data on the $i$th worker. The servers will collect and aggregate the gradients sent by the workers, and update the model as follows:
\begin{align*}
x^{t+1} = x^t - \gamma^t Aggr(\{\tilde{v}_i^t: i \in [m]\}),
\end{align*}
where $Aggr(\cdot)$ is an aggregation rule (e.g., averaging), and $\{\tilde{v}_i^t: i \in [m]\}$ is the set of gradient estimators received by the servers in the $t^{\mbox{th}}$ iteration. Under Byzantine failures/attacks, $\{v_i^t = \nabla F_i(x^t): i \in [m]\}$ is partially replaced by arbitrary values, which yields $\{\tilde{v}_i^t: i \in [m]\}$.

%

\section{Byzantine resilience}

In this section, we formally define the classic Byzantine resilience property and its generalized version: dimensional Byzantine resilience.

Suppose that in a specific iteration, the correct vectors $\{v_i: i \in [m]\}$ are i.i.d samples drawn from the random variable $G = \sum_{j=1}^n \nabla f(x; z^j)$, where $\E[G] = g$ is an unbiased estimator of the gradient based on the current parameter $x$. Thus, $\E[v_i] = \E[G] = g$, for any $i \in [m]$. We simplify the notations by ignoring the index of iteration $t$.

We first introduce the classic Byzantine model, which is reformulated from the model proposed by \citet{blanchard2017machine}.
With the Byzantine workers, the vectors $\{\tilde{v}_i: i \in [m]\}$ which are actually received by the server nodes are as follows: 
\begin{definition}[Classic Byzantine Model]
\begin{align}
\tilde{v}_i = 
\begin{cases}
v_i, &\mbox{if the $i$th worker is correct,}\\
arbitrary, &\mbox{if the $i$th worker is Byzantine}.
\end{cases}
\label{equ:byz_worker}
\end{align}
\end{definition}
Note that the indices of Byzantine workers can change across different iterations. Furthermore, the server nodes are not aware of which workers are Byzantine. The only information given is the number of Byzantine workers, if necessary.

We then introduce the classic Byzantine resilience.
\begin{definition} 
\label{def:byz}
(Classic $\Delta$-Byzantine Resilience). Assume that $0 \leq q \leq m$. Let $\{v_i: i \in [m]\}$ be any i.i.d. random vectors in $\R^d$, $v_i \sim G$, with $\E[G] = g$. Let $\{\tilde{v}_i: i \in [m]\}$ be the set of vectors, of which up to $q$ of them are replaced by arbitrary vectors in $\R^d$, while the others still equal to the corresponding $\{v_i\}$. Aggregation rule $\aggr(\cdot)$ is said to be classic $\Delta$-Byzantine resilient if 
$
\E \|\aggr(\{\tilde{v}_i: i \in [m]\}) - g \|^2 \leq \Delta,
$
where $\Delta$ is a constant dependent on $m$ and $q$.
\end{definition}

The baseline algorithm \texttt{Krum} is defined as follows. 
\begin{definition} 
\label{def:krum}
Krum chooses the vector with the minimal local sum of distances: 
$
\krum (\{\tilde{v}_i: i \in [m]\}) = \tilde{v}_k, \quad
k = \argmin_{i \in [m]} \sum_{i \rightarrow j} \| \tilde{v}_i - \tilde{v}_j \|^2,
$
where $i \rightarrow j$ is the indices of the $m-q-2$ nearest neighbours of $\tilde{v}_i$ in $\{\tilde{v}_i: i \in [m]\}$ measured by Euclidean distance.
\end{definition}
The \texttt{Krum} aggregation is classic $\Delta$-Byzantine resilient under certain assumptions. The proof is given by Proposition 1 of \citet{blanchard2017machine}. 
\begin{lemma}[\citet{blanchard2017machine}]
Let $v_1, \ldots, v_m$ be any i.i.d. random $d$-dimensional vectors s.t. $v_i \sim G$, with $\E[G] = g$ and $\E\|G-g\|^2 \leq V$. $q$ of $\{\tilde{v}_i: i \in [m]\}$ are Byzantine. If $2q + 2 < m$, we have 
$
\E \|\krum(\{\tilde{v}_i: i \in [m]\}) - g \|^2 \leq \Delta_0,
$
where 
$
\Delta_0 = \left(6m-6q + \frac{4q(m-q-2) + 4q^2(m-q-1)}{m-2q-2} \right)V.
$
\end{lemma}

The generalized Byzantine model is denoted as:
\begin{definition}[Generalized Byzantine Model]
\begin{align}
(\tilde{v}_i)_j = 
\begin{cases}
(v_i)_j, &\mbox{if the the $j$th dimension of $v_i$ is correct,}\\
arbitrary, &\mbox{otherwise},
\end{cases}
\label{equ:byz_model}
\end{align}
where $(v_i)_j$ is the $j$th dimension of the vector $v_i$.
\end{definition}

Based on the Byzantine model above, we introduce a generalized Byzantine resilience property, dimensional $\Delta$-Byzantine resilience, which is defined as follows:
\begin{definition} 
\label{def:dim_byz}
(Dimensional $\Delta$-Byzantine Resilience). Assume that $0 \leq q \leq m$.  Let $\{v_i: i \in [m]\}$ be any i.i.d. random vectors in $\R^d$, $v_i \sim G$, with $\E[G] = g$. Let $\{\tilde{v}_i: i \in [m]\}$ be the set of candidate vectors. For each dimension, up to $q$ of the $m$ values are replaced by arbitrary values, i.e., for dimension $j \in [d]$, $q$ of $\{(\tilde{v}_i)_j: i \in [m]\}$ are Byzantine, where $(\tilde{v}_i)_j$ is the $j$th dimension of the vector $\tilde{v}_i$. Aggregation rule $\aggr(\cdot)$ is said to be dimensional $\Delta$-Byzantine resilient if 
$
\E \|\aggr(\{\tilde{v}_i: i \in [m]\}) - g \|^2 \leq \Delta,
$
where $\Delta$ is a constant dependent on $m$ and $q$.
\end{definition}

Note that classic $\Delta$-Byzantine resilience is a special case of dimensional $\Delta$-Byzantine resilience. For classic Byzantine resilience defined in Definition~\ref{def:byz}, all the Byzantine values must lie in the same subset of workers, as shown in Figure~\ref{fig:viz_byz}(a).

In the following propositions, we show that \texttt{Mean} and \texttt{Krum} are not dimensional Byzantine resilient~($\E \|\aggr(\{\tilde{v}_i: i \in [m]\}) - g \|^2$ is unbounded). The proofs are provided in the appendix.
\begin{proposition}
\label{thm:mean_dim_byz}
The averaging aggregation rule is not dimensional Byzantine-resilient.
\end{proposition}

\begin{proposition}
\label{thm:dim_byz}
Any aggregation rule $\aggr(\{\tilde{v}_i: i \in [m]\})$ that outputs $\aggr \in \{\tilde{v}_i: i \in [m]\}$ is not dimensional Byzantine resilient.
\end{proposition}
\texttt{Krum} chooses the vector with the minimal score, which is not dimensional Byzantine-resilient.
\begin{proposition}
\label{lem:krum_dim_byz}
$\krum(\cdot)$ is not dimensional Byzantine-resilient.
\end{proposition}

\section{Trimmed-mean-based aggregation}
\label{sec:median}

With the Byzantine failure model defined in Equation~(\ref{equ:byz_worker}) and (\ref{equ:byz_model}), we propose two trimmed-mean-based aggregation rules, which are Byzantine resilient under certain conditions. 

%
%

\subsection{Trimmed mean}
To define the trimmed mean, we first define the order statistics.
\begin{definition}(Order Statistics)
\label{def:ord_stat}
By sorting the scalar sequence $\{u_i: i \in [m]\}$, we get $u_{1:m} \leq u_{2:m} \leq \ldots \leq u_{m:m}$, where $u_{k:m}$ is the $k$th smallest element in $\{u_i: i \in [m]\}$.
\end{definition}
Then, we define the trimmed mean.
\begin{definition}(Trimmed Mean)
\label{def:trim}
For $b \in \mathbb{Z} \cap [0, \lceil m/2 \rceil - 1]$, the $b$-trimmed mean of the set of scalars $\{u_i: i \in [m]\}$ is defined as follows:
\[
\trmean_b(\{u_i: i \in [m]\}) = \frac{1}{m-2b} \sum_{k=b+1}^{m-b} u_{k:m},
\] where  $u_{k:m}$ is the $k$th smallest element in $\{u_i: i \in [m]\}$ defined in Definition~\ref{def:ord_stat}. The high-dimensional version, $Trmean_b(\{\tilde{v}_i: i \in [m]\})$, simply applies $\trmean_b(\cdot)$ in the coordinate-wise manner.
\end{definition}

The following theorem claims that by using $\trmean_b(\cdot)$, the resulting vector is dimensional Byzantine resilient. A proof is provided in the appendix.

\begin{theorem}(Bounded Variance)
\label{thm:bound_trmean_var}
Let $v_1, \ldots, v_m$ be any i.i.d. random $d$-dimensional vectors s.t. $v_i \sim G$, with $\E[G] = g$ and $\E\|G-g\|^2 \leq V$. In each dimension, $q$ values are Byzantine, which yields $\{\tilde{v}_i: i \in [m]\}$. If $2q < m$, we have 
$
\E \|\trmean_b(\{\tilde{v}_i: i \in [m]\}) - g \|^2 \leq \Delta_1,
$
where 
$
\Delta_1 = \frac{2(b+1)(m-q)}{(m-b-q)^2} V.
$
\end{theorem}
Theorem~\ref{thm:bound_trmean_var} tells us that the upper bound of the variance $\E \left\| \trmean_b(\{\tilde{v}_i: i \in [m]\}) - g \right\|^2$ decreases when $m$ increases, $b$ decreases, $q$ decreases, or $V$ decreases.


\subsection{Beyond trimmed mean}
Using the trimmed mean, we have to drop $2b$ elements for each dimension. In this section, we explore the possibility of aggregating more elements. 
To be more specific, for each dimension, we take the average of the $m-b$ values nearest to the trimmed mean. We call the resulting aggregation rule \texttt{Phocas}~\footnote{The name of a Byzantine emperor.}, which is defined as follows:
\begin{definition}(Phocas)
\label{def:phocas}
We sort the scalar sequence $\{u_i: i \in [m]\}$ by using the distance to a certain value $y$:
$
|u_{1/y} - y| \leq |u_{2/y} - y| \leq \ldots \leq |u_{m/y} - y|, 
$
where $u_{k/y}$ is the $k$th nearest element to $y$ in $\{u_i: i \in [m]\}$. 
\textit{Phocas} is the average of the first $(m-b)$ nearest elements to the $b$-trimmed mean $Trmean_b = \trmean_b(\{u_i: i \in [m]\})$:
\begin{align*}
\phocas_b(\{u_i: i \in [m]\}) = \frac{\sum_{i=1}^{m-b} u_{i/Trmean_b}}{m-b}.
\end{align*}
The high-dimensional version, $\phocas_b(\{\tilde{v}_i: i \in [m]\})$, simply applies $\phocas_b(\cdot)$ in the coordinate-wise manner.
\end{definition}

We show that $\phocas(\cdot)$ is dimensional Byzantine-resilient.
\begin{theorem}(Bounded Variance)
\label{thm:bound_phocas_var}
Let $v_1, \ldots, v_n$ be any i.i.d. random $d$-dimensional vectors s.t. $v_i \sim G$, with $\E[G] = g$ and $\E\|G-g\|^2 \leq V$. In each dimension, $q$ values are Byzantine, which yields $\{\tilde{v}_i: i \in [m]\}$ If $2q < n$, we have 
$
\E \|\phocas_b(\{\tilde{v}_i: i \in [m]\}) - g \|^2 \leq \Delta_2,
$
where 
$
\Delta_2 = \left[4 + \frac{12(b+1)(m-q)}{(m-b-q)^2} \right] V.
$
\end{theorem}
The \texttt{Phocas} aggregation can be viewed as a trimmed average centering at the trimmed mean, which filters out the values far away from the trimmed mean. Similar to the trimmed mean, the variance of \texttt{Phocas} decreases when $m$ increases, $b$ decreases, $q$ decreases, or $V$ decreases.

\subsection{Convergence analysis}
In this section, we provide the convergence guarantees for synchronous SGD with $\Delta$-Byzantine-resilient aggregation rules. The proofs can be found in the appendix.
We first introduce the two conditions necessary in our convergence analysis.
\begin{definition}
If $F(x)$ is $L_F$-smooth, then 
$
F(y) - F(x) \leq \ip{\nabla F(x)}{y-x} + \frac{L_F}{2} \|y-x\|^2, \forall x,y \in \R^d,
$
where $L_F \geq 0$.
If $F(x)$ is $\mu_F$-strongly convex, then 
$
\ip{\nabla F(x)}{y-x} + \frac{\mu_F}{2} \|y-x\|^2 \leq F(y) - F(x), \forall x,y \in \R^d,
$
where $\mu_F \geq 0$.
\end{definition}
First, we prove that for strongly convex and smooth loss functions, SGD with $\Delta$-Byzantine-resilient aggregation rules has linear convergence with a constant error.
\begin{theorem}
Assume that $F(x)$ is $\mu_F$-strongly convex and $L_F$-smooth, where $0 < \mu_F \leq L_F$. We take $\gamma \leq \frac{2}{\mu_F + L_F}$. In any iteration $t$, the correct gradients are $v_i^t = \nabla F_i(x^t)$. Using any (classic or dimensional) $\Delta$-Byzantine-resilient aggregation rule with corresponding assumptions, we obtain linear convergence with a constant error after $T$ iterations with synchronous SGD:
\begin{align*}
\E\|x^{T} - x^*\| \leq \left( 1- \frac{\gamma \mu_F L_F}{\mu_F + L_F} \right)^T \|x^0-x^*\| + \frac{\mu_F + L_F}{\mu_F L_F} \gamma  \sqrt{\Delta}.
\end{align*}
\end{theorem}

Then, we prove the convergence of SGD for general smooth loss functions.
\begin{theorem}
Assume that $F(x)$ is $L_F$-smooth and potentially non-convex, where $0 < L_F$. We take $\gamma \leq \frac{1}{L_F}$. In any iteration $t$, the correct gradients are $v_i^t = \nabla F_i(x^t)$. Using any (classic or dimensional) $\Delta$-Byzantine-resilient aggregation rule with corresponding assumptions, we obtain linear convergence with a constant error after $T$ iterations with synchronous SGD:
\begin{align*}
\frac{\sum_{i=0}^{T-1} \E\|\nabla F(x^i)\|^2}{T} \leq \frac{2}{\gamma T} \left[ F(x^0) - F(x^*) \right] + \Delta.
\end{align*}
\end{theorem}

\subsection{Time complexity}
For the trimmed mean, we only need to find the order statistics of each dimension. To do so, we use the so-called \textit{selection algorithm}~\cite{blum1973time} with linear time complexity to find the $k$th smallest element. In general, the time complexity is $O(d(m-2b)m)$. When $b$ is large, the factor $m-2b \ll m$ can be ignored, which yields the nearly linear time complexity $O(dm)$. When $b$ is small, the time complexity is the same as the sorting algorithm, which is $O(dm\log m).$ For \texttt{Phocas}, the computation additional to computing the trimmed takes linear time $O(dm)$. Thus, the time complexity is the same as \texttt{Trmean}.
Note that for \texttt{Krum} and \texttt{Multi-Krum}, the time complexity is $O(dm^2)$~\cite{blanchard2017machine}.

\begin{table*}[htb!]
\caption{Experiment Summary}
\label{table:datasets}
\begin{center}
\begin{tabular}{|l|r|r|r|r|r|r|r|r|}
\hline 
Dataset & \# train & \# test  & $\gamma$ & \# rounds & Batchsize & Evaluation metric   \\ \hline 
MNIST~\cite{loosli2007training} & 60k & 10k & 0.1 & 500 & 32 & top-1 accuracy \\ \hline 
CIFAR10~\cite{krizhevsky2009learning} & 50k & 10k & 5e-4 & 4000 & 128 & top-3 accuracy \\ \hline 
\end{tabular} 
\end{center}
\end{table*}

\begin{figure*}[htb!]
\centering
\subfigure[Gaussian attack. 6 out of 20 gradient vectors are replaced by i.i.d. random vectors drawn from a Gaussian distribution with 0 mean and 200 standard deviation.]{\includegraphics[width=0.48\textwidth]{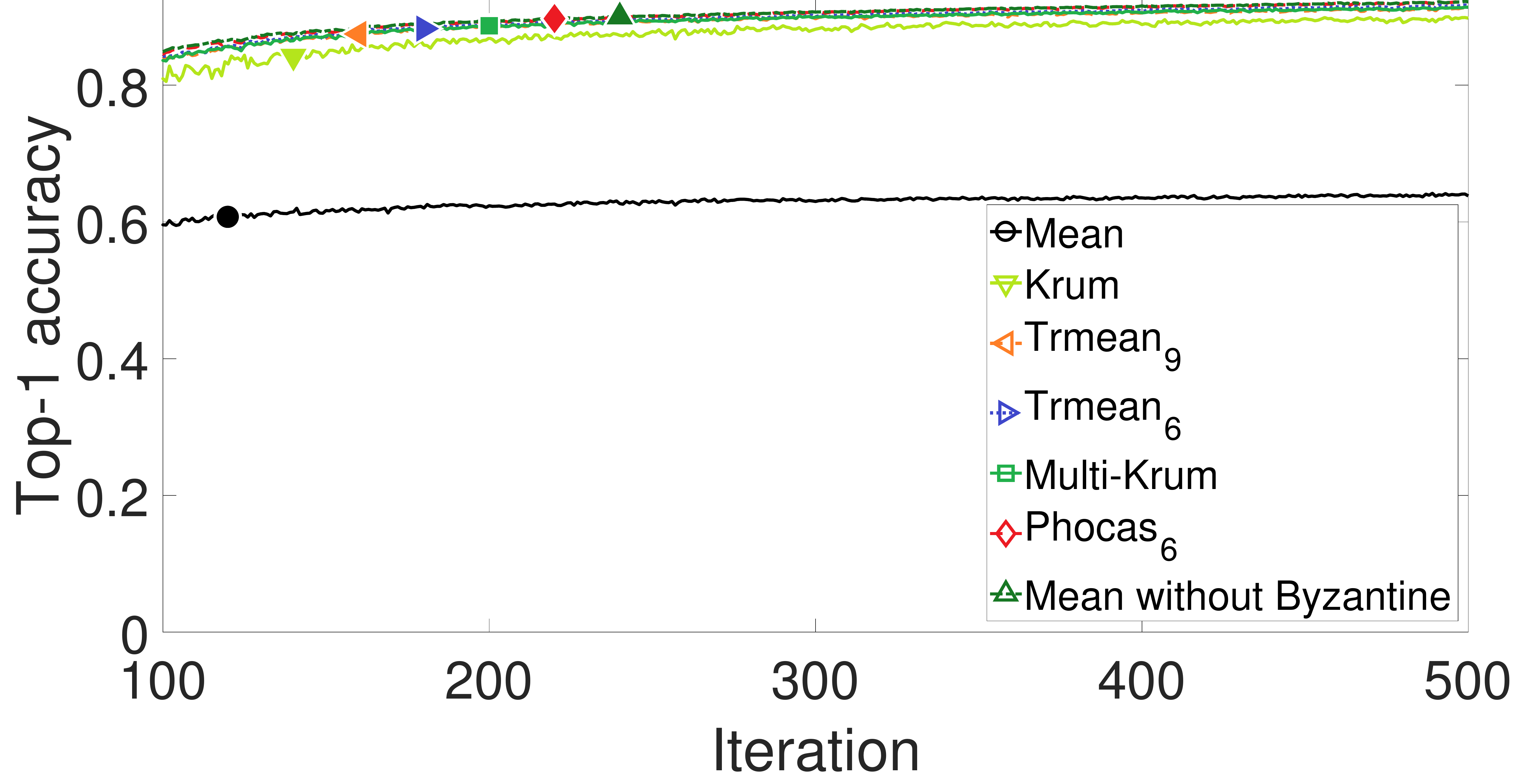}}\hspace{10pt}
\subfigure[Omniscient attack. 6 out of 20 gradient vectors are replaced by the negative sum of all the correct gradients, scaled by a large constant~(1e20).]
{\includegraphics[width=0.48\textwidth]{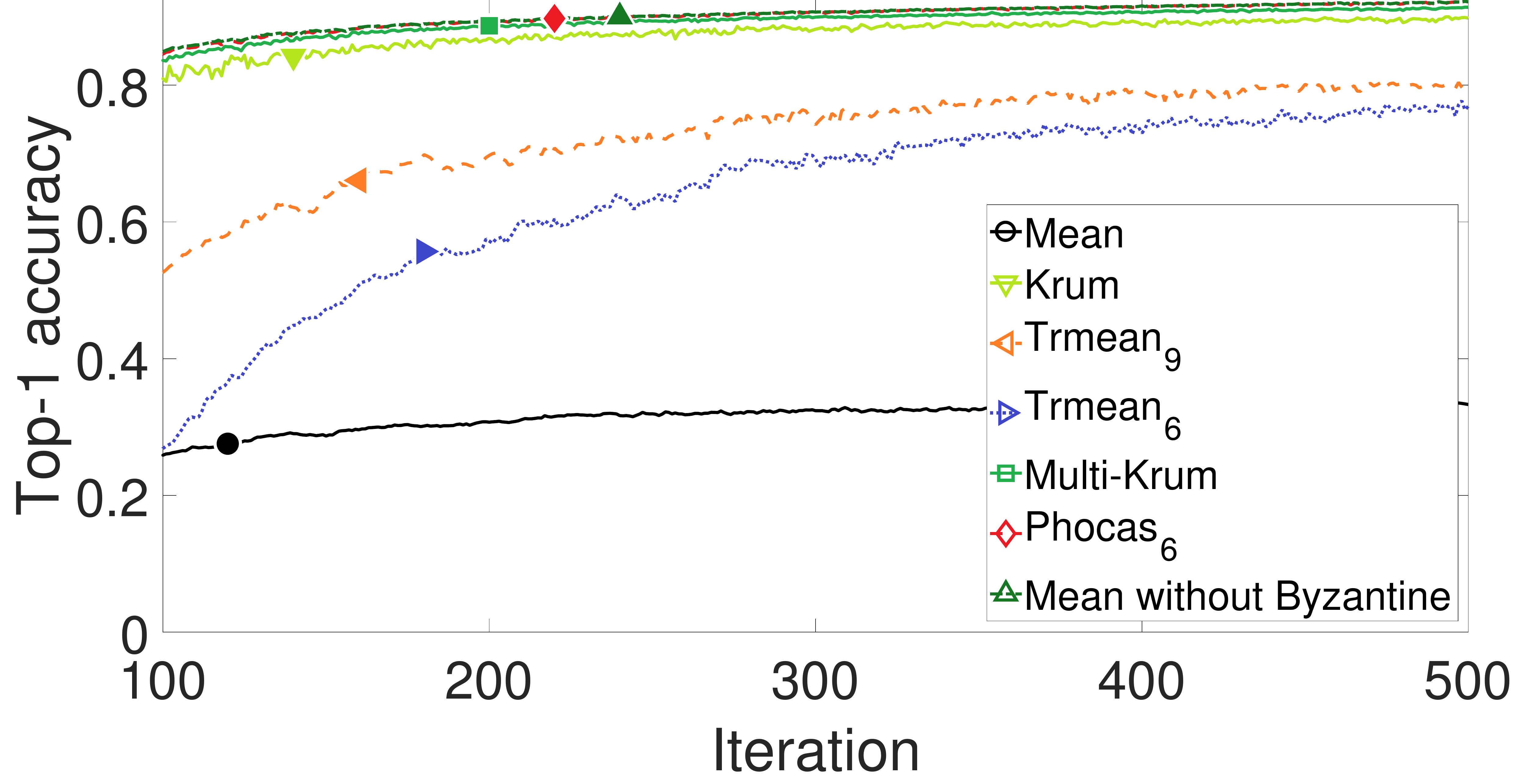}}
\subfigure[Bit-flip attack. For each of the first 1000 dimensions, 1 of the 20 floating numbers is manipulated by flipping the 22th, 30th, 31th and 32th bits.]{\includegraphics[width=0.48\textwidth]{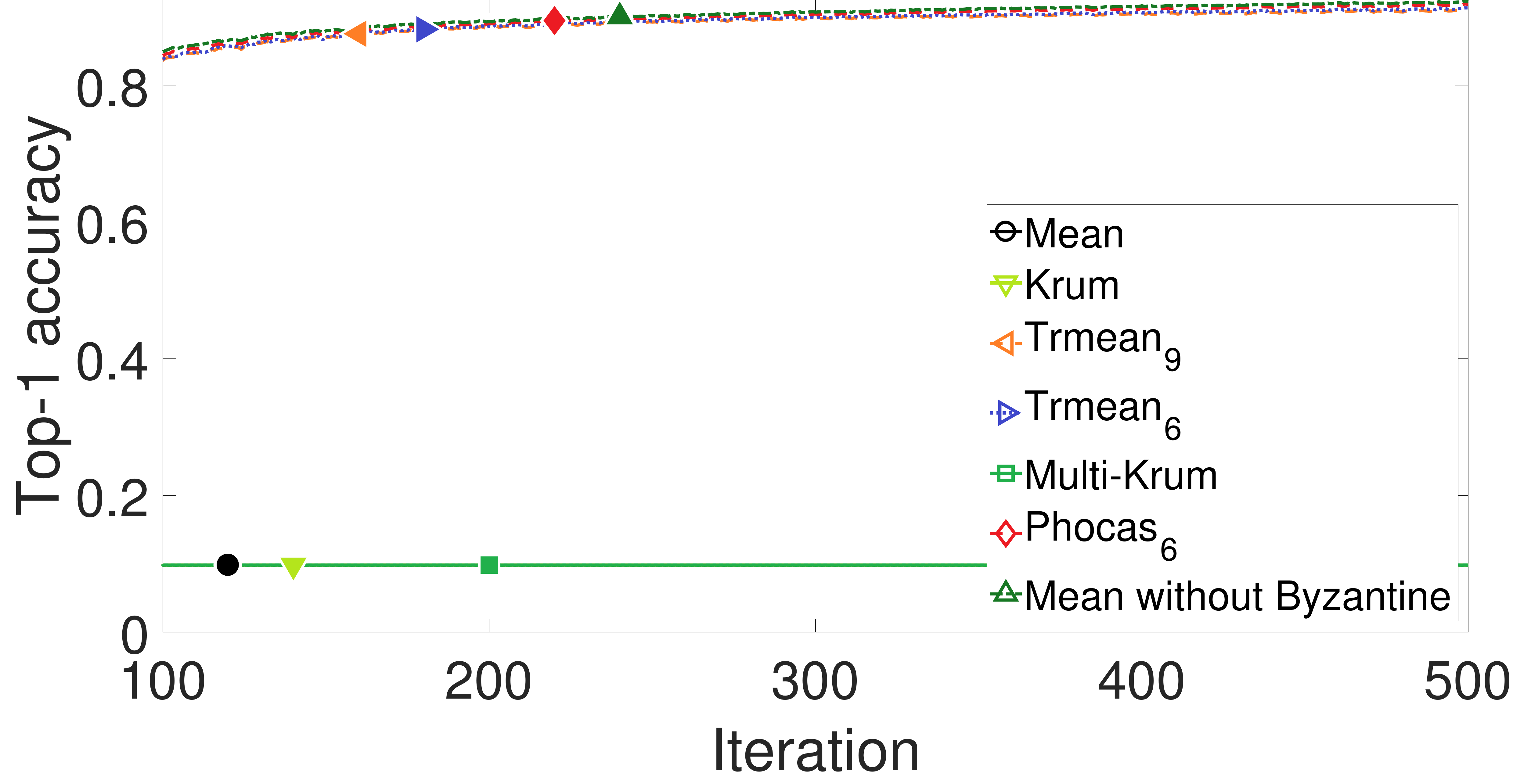}}\hspace{10pt}
\subfigure[Gambler attack. The parameters are evenly assigned to 20 servers. For one single server, any received value is multiplied by $-1e20$ with probability 0.05\%.]{\includegraphics[width=0.48\textwidth]{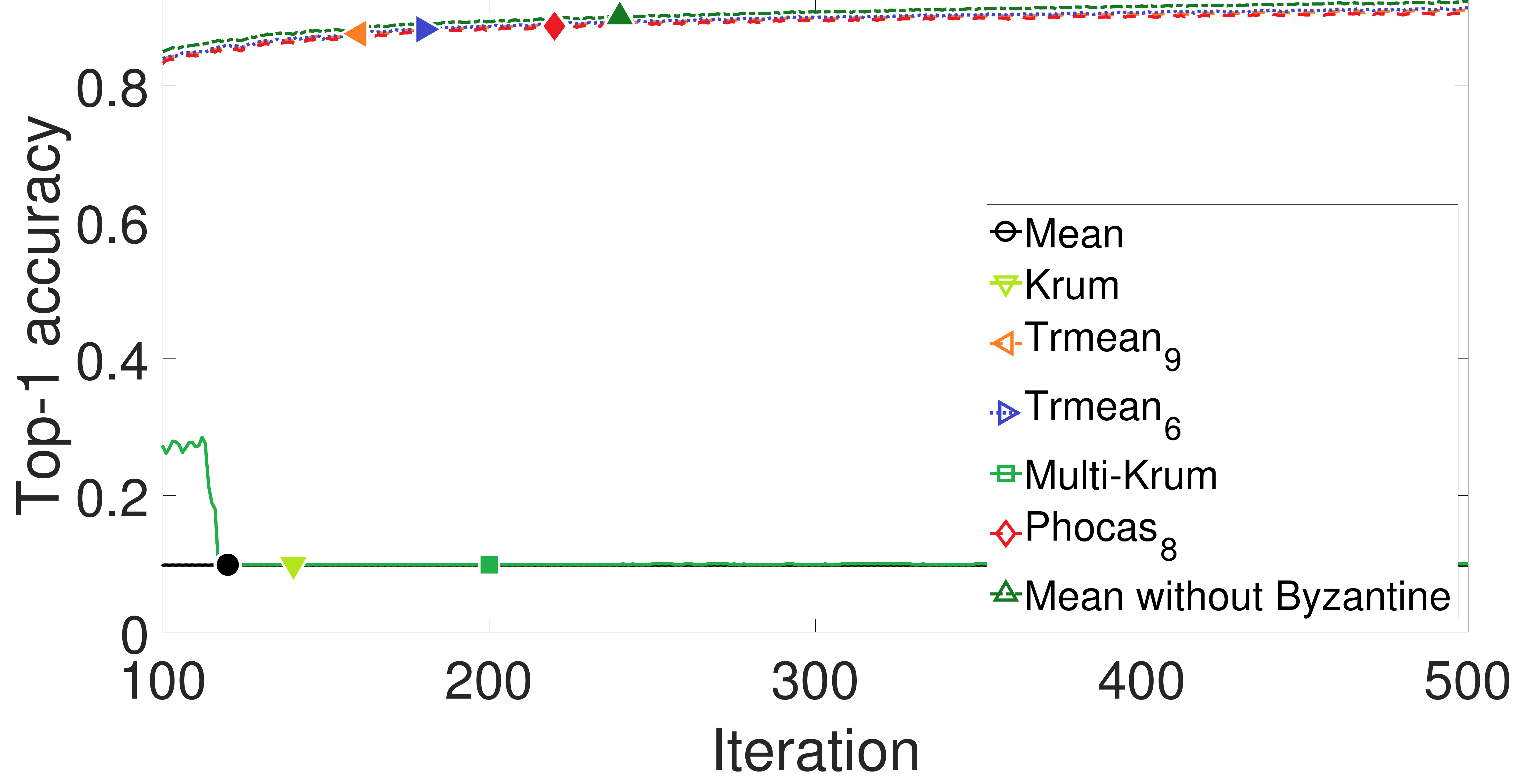}}
\caption{Top-1 accuracy of MLP on MNIST with different attacks. }
\label{fig:mnist_mlp}
\end{figure*}
\begin{figure*}[htb!]
\vspace{-0.3cm}
\centering
\subfigure[Accuracy of $Krum$-based aggregations under bit-flip attack, at the end of training, when $q$ varies.]{\includegraphics[width=0.48\textwidth,height=0.12\textheight]{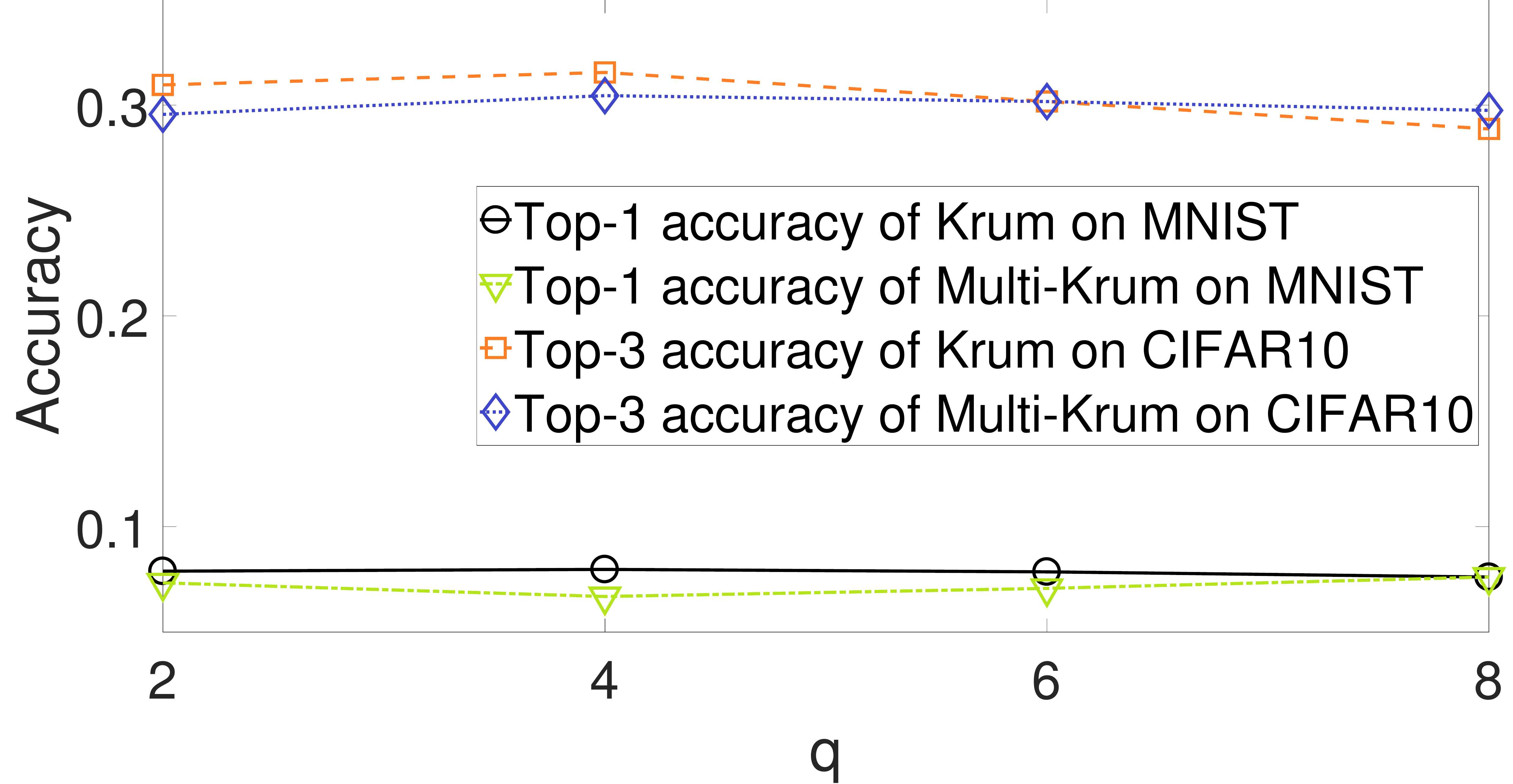}}\hspace{10pt}
\subfigure[Maximal accuracy under gambler attack throughout training, when $b$~($q$ for Krum and Multi-Krum) varies.]
{\includegraphics[width=0.48\textwidth,height=0.12\textheight]{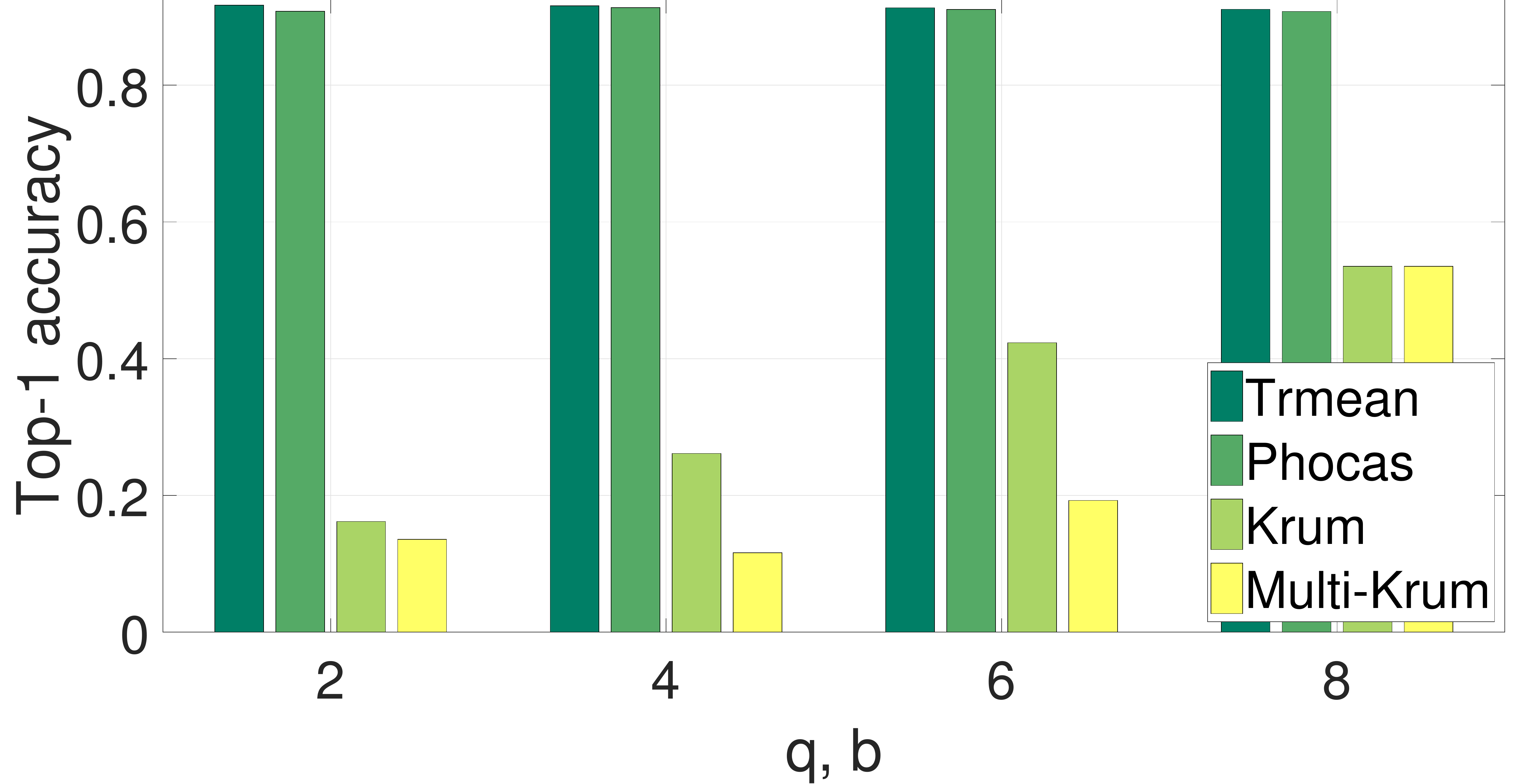}}
\caption{Sensitivity to hyperparameters.}
\label{fig:sensitivity}
\vspace{-0.3cm}
\end{figure*}

\section{Experiments}
\label{sec:experiments}
In this section, we evaluate the Byzantine resilience properties of the proposed algorithms. We consider two image classification tasks: handwritten digits classification on MNIST dataset using multi-layer perceptron~(MLP), and object recognition on convolutional neural network~(CNN). The details of these two neural networks can be found in the appendix. There are $m=20$ worker processes. We repeat each experiment for ten times and take the average. 
The details of the datasets and the default hyperparameters of the corresponding models are listed in Table~\ref{table:datasets}. We use top-1 or top-3 accuracy on testing sets~(disjoint from the training sets) as evaluation metrics.

The baseline aggregation rules are \texttt{Mean}, \texttt{Krum}~(Definition~\ref{def:krum}), and \texttt{Multi-Krum}. 
\texttt{Multi-Krum} is a variant of \texttt{Krum} defined in \citet{blanchard2017machine}, which takes the average on several vectors selected by multiple rounds of \texttt{Krum}. We also include the averaging without Byzantine failures as a baseline, which is referred to as \texttt{Mean without Byzantine}.
We compare these baseline algorithms with the proposed algorithms: \texttt{Trmean} defined in Definition~\ref{def:trim}, and \texttt{Phocas} defined in Definition~\ref{def:phocas}, under different attacks.

Note that all the experiments of CNN on CIFAR10 show similar results with the experiments of MLP on MNIST. Thus, we put the results of CNN in the appendix.

%

\subsection{Byzantine resilience}
In this section, we test the Byzantine resilience of the proposed algorithms under different kinds of attacks. The zoomed figure of each experiment can be found in the appendix.

\subsubsection{Gaussian attack}
We test classic Byzantine resilience in this experiment. 
We consider the attackers that replace some of the gradient vectors with Gaussian random vectors with zero mean and isotropic covariance matrix with standard deviation 200. We refer to this kind of attack as \textit{Gaussian attack}. 6 out of the 20 gradient vectors are Byzantine. The results are shown in Figure~\ref{fig:mnist_mlp}(a).
As expected, averaging is not Byzantine resilient. The gaps between all the other algorithms are tiny. \texttt{Phocas} performs like there are no Byzantine failures at all. \texttt{Krum}, \texttt{Multi-Krum}, and \texttt{Trmean} converge slightly slower.  

\subsubsection{Omniscient attack}
We test classic Byzantine resilience in this experiment. 
This kind of attacker is assumed to know the all the correct gradients. For each Byzantine gradient vector, the gradient is replaced by the negative sum of all the correct gradients, scaled by a large constant~(1e20 in the experiments). Roughly speaking, this attack tries to make the parameter server go into the opposite direction with a long step. 6 out of the 20 gradient vectors are Byzantine. The results are shown in Figure~\ref{fig:mnist_mlp}(b). \texttt{Phocas} still performs just like there is no failure. \texttt{Multi-Krum} is not as good as \texttt{Phocas}, but the gap is small.  \texttt{Krum} converges slower. However, \texttt{Trmean} converges to bad solutions. 

\subsubsection{Bit-flip attack}
We test dimensional Byzantine resilience in this experiment. 
Knowing the information of other workers can be difficult in practice. Thus, we use more realistic scenario in this experiment. The attacker only manipulates some individual floating numbers by flipping the 22th, 30th, 31th and 32th bits. Furthermore, we test dimensional Byzantine resilience in this experiment. For each of the first 1000 dimensions, 1 of the 20 floating numbers is manipulated using the bit-flip attack. The results are shown in Figure~\ref{fig:mnist_mlp}(c). As expected, only \texttt{Phocas} and \texttt{Trmean} are dimensional Byzantine resilient.

Note that for \texttt{Krum} and \texttt{Multi-Krum}, their assumption requires the number of Byzantine vectors $q$ to satisfy $2q + 2 < m$, which means $q \leq 8$ in our experiments. However, because each gradient is partially manipulated, all the $m$ vectors are Byzantine, which breaks the assumption of the Krum-based algorithms. Furthermore, to compute the distances to the $(m-q-2)$-nearest neighbours, $m-q-2$ must be positive. To test the performance of \texttt{Krum} and \texttt{Multi-Krum}, we set $q = 8$ for these two algorithms so that they can still be executed. 
Furthermore, we test whether tuning $q$ can make a difference. The results are shown in Figure~\ref{fig:sensitivity}(a). Obviously, whatever $q$ we use, \texttt{Krum}-based algorithms get stuck around bad solutions. 

\subsubsection{General attack with multiple servers}
\label{sec:general_attack}
We test general Byzantine resilience in this experiment. 
We evaluate the robust aggregation rules under a more general and realistic type of attack. It is very popular to partition the parameters into disjoint subsets, and use multiple server nodes to storage and aggregate them~\citep{Li2014ScalingDM,Li2014CommunicationED,Ho2013MoreED}. We assume that the parameters are evenly partitioned and assigned to the server nodes. The attacker picks one server, and manipulates any floating number by multiplying $-1e20$, with probability of $0.05\%$. We call this attack \textit{gambler}, because the attacker randomly manipulates the values, with the goal that in some iterations the assumptions/prerequisites of the robust aggregation rules are broken, which crashes the training. Such an attack requires less global information, and can be concentrated on one single server, which makes it more realistic and easier to implement.

In Figure~\ref{fig:mnist_mlp}(d), we evaluate the performance of all the robust aggregation rules under the gambler attack. The number of servers is $20$. For \texttt{Krum}, \texttt{Multi-Krum} and \texttt{Phocas}, we set the estimated Byzantine number $q=b=8$. Only \texttt{Phocas} and \texttt{Trmean} survive under this attack. The convergence is slightly slower than the averaging without Byzantine values, but the gaps are small. 

\subsubsection{Sensitivity to the hyperparameters}
We test the robustness to the estimated number of Byzantine workers $b$~($q$ for \texttt{Krum} and \texttt{Multi-Krum}) in this experiment. We show the maximal accuracy throughout the training. The results are shown in Figure~\ref{fig:sensitivity}(b). The performance of \texttt{Phocas} and \texttt{Trmean} does not significantly change when $b$ varies. 

\subsection{Discussion}
As expected, \texttt{Mean} aggregation is not Byzantine resilient. \texttt{Krum}, \texttt{Multi-Krum} are classic Byzantine-resilient but not dimensional Byzantine-resilient. \texttt{Phocas} and \texttt{Trmean} are dimensional Byzantine resilient. However, under omniscient attack, \texttt{Trmean} suffers from larger variances, which slow down the convergence. 

The gambler attack shows the true advantage of dimensional Byzantine resilience: higher probability of survival. Under such attack, chances are that the assumptions/prerequisites of \texttt{Phocas} and \texttt{Trmean} may still fail. However, their probability of crashing is less than the other algorithms because dimensional Byzantine resilience generalizes classic Byzantine resilience. An interesting observation is that \texttt{Trmean} is slightly better than \texttt{Phocas} under gambler attack. That is because the estimation of $q = b = 8$ is not accurate, which will cause some unpredictable behavior for \texttt{Phocas}. We choose $q = b = 8$ because it is the maximal value we can take for \texttt{Krum} and \texttt{Multi-Krum}. 

It is obvious that \texttt{Phocas} performs best in almost cases. \texttt{Multi-Krum} is also good, except that it is not dimensional Byzantine-resilient. The reason why \texttt{Phocas} and \texttt{Multi-Krum} have better performance is that they aggregate more candidates to stabilize the convergence. Note that \texttt{Phocas} not only performs just as well as or even better than \texttt{Multi-Krum}, but also has lower time complexity. 

\texttt{Trmean} has the cheapest computation. Its worst case, omniscient attack, is hard to implement in reality. Thus, for most applications, we suggest \texttt{Trmean} as an easy-to-implement aggregation rule with robust performance. However, if we assume the worst cases of the attacks/failures, \texttt{Phocas} should be adopted for best robustness.

\vspace{-0.15cm}
\section{Related work}

Our work is closely related to \citet{blanchard2017machine} and \citet{yin2018byzantine}. Another paper~\cite{chen2017distributed} proposed grouped geometric median for Byzantine resilience, with strongly convex functions. 

Our approach offers the following important advantages over the previous work.
\setitemize[0]{leftmargin=*}
\begin{itemize}
\item \textbf{Cheaper computation compared to Krum.} \texttt{Trmean} and \texttt{Phocas} have nearly linear time complexity, while the time complexity of Krum is $O(m^2 d)$.
\item \textbf{Dimensional Byzantine resilience.} \texttt{Trmean} and \texttt{Phocas} tolerate a more general type of Byzantine failures described in Equation~(\ref{equ:byz_model}) and Definition~\ref{def:dim_byz}, while \texttt{Krum} can only tolerate the classic Byzantine failures described in Equation~(\ref{equ:byz_worker}) and Definition~\ref{def:byz}.
\item \textbf{Simpler dimension-free convergence guarantees with fewer assumptions.} \citet{yin2018byzantine} also study the Byzantine resilience of the trimmed mean and its special case, median. However, in that work, the bounds grow with the number of dimensions $d$, even if the variance of gradients $V$ is fixed. To establish the bounds, the authors assume bounded domain, and sub-exponential gradients with bounded skewness, which are not required in our theoretical analysis. In this paper, we use fewer assumptions to prove the dimension-free theoretical guarantees of trimmed mean. 
\end{itemize}

The major contribution of this paper is a combination of theory and practice. First, we provide the theoretical guarantee of the convergence of the trimmed mean with fewer assumptions. Then, we propose a novel aggregation rule, \texttt{Phocas}, which has comparable theoretical guarantees, and comparable or even better performance in the experiments.

\section{Conclusion}
We investigate the generalized Byzantine resilience, and propose trimmed-mean-based aggregation rules for synchronous SGD. The algorithms have low time complexity and provable convergence. Our empirical results show good performance. We will study the Byzantine resilience in other scenarios such as asynchronous training in the future work.

\bibliography{byz}

\begin{thebibliography}{18}
\providecommand{\natexlab}[1]{#1}
\providecommand{\url}[1]{\texttt{#1}}
\expandafter\ifx\csname urlstyle\endcsname\relax
  \providecommand{\doi}[1]{doi: #1}\else
  \providecommand{\doi}{doi: \begingroup \urlstyle{rm}\Url}\fi

\bibitem[Alistarh et~al.(2018)Alistarh, Allen-Zhu, and
  Li]{alistarh2018byzantine}
D.~Alistarh, Z.~Allen-Zhu, and J.~Li.
\newblock Byzantine stochastic gradient descent.
\newblock \emph{arXiv preprint arXiv:1803.08917}, 2018.

\bibitem[Blanchard et~al.(2017)Blanchard, Guerraoui, Stainer,
  et~al.]{blanchard2017machine}
P.~Blanchard, R.~Guerraoui, J.~Stainer, et~al.
\newblock Machine learning with adversaries: Byzantine tolerant gradient
  descent.
\newblock In \emph{Advances in Neural Information Processing Systems}, pages
  118--128, 2017.

\bibitem[Blum et~al.(1973)Blum, Floyd, Pratt, Rivest, and Tarjan]{blum1973time}
M.~Blum, R.~W. Floyd, V.~Pratt, R.~L. Rivest, and R.~E. Tarjan.
\newblock Time bounds for selection.
\newblock \emph{Journal of computer and system sciences}, 7\penalty0
  (4):\penalty0 448--461, 1973.

\bibitem[Bubeck et~al.(2015)]{bubeck2015convex}
S.~Bubeck et~al.
\newblock Convex optimization: Algorithms and complexity.
\newblock \emph{Foundations and Trends{\textregistered} in Machine Learning},
  8\penalty0 (3-4):\penalty0 231--357, 2015.

\bibitem[Chen et~al.(2017)Chen, Su, and Xu]{chen2017distributed}
Y.~Chen, L.~Su, and J.~Xu.
\newblock Distributed statistical machine learning in adversarial settings:
  Byzantine gradient descent.
\newblock \emph{arXiv preprint arXiv:1705.05491}, 2017.

\bibitem[Dean et~al.(2012)Dean, Corrado, Monga, Chen, Devin, Le, Mao, Ranzato,
  Senior, Tucker, Yang, and Ng]{Dean2012LargeSD}
J.~Dean, G.~S. Corrado, R.~Monga, K.~Chen, M.~Devin, Q.~V. Le, M.~Z. Mao,
  M.~Ranzato, A.~W. Senior, P.~A. Tucker, K.~Yang, and A.~Y. Ng.
\newblock Large scale distributed deep networks.
\newblock In \emph{NIPS}, 2012.

\bibitem[Harinath et~al.(2017)Harinath, Satyanarayana, and
  Murthy]{harinath2017review}
D.~Harinath, P.~Satyanarayana, and M.~R. Murthy.
\newblock A review on security issues and attacks in distributed systems.
\newblock \emph{Journal of Advances in Information Technology}, 8\penalty0 (1),
  2017.

\bibitem[Ho et~al.(2013)Ho, Cipar, Cui, Lee, Kim, Gibbons, Gibson, Ganger, and
  Xing]{Ho2013MoreED}
Q.~Ho, J.~Cipar, H.~Cui, S.~Lee, J.~K. Kim, P.~B. Gibbons, G.~A. Gibson, G.~R.
  Ganger, and E.~P. Xing.
\newblock More effective distributed ml via a stale synchronous parallel
  parameter server.
\newblock \emph{Advances in neural information processing systems},
  2013:\penalty0 1223--1231, 2013.

\bibitem[Kingma and Ba(2014)]{Kingma2014AdamAM}
D.~P. Kingma and J.~Ba.
\newblock Adam: A method for stochastic optimization.
\newblock \emph{CoRR}, abs/1412.6980, 2014.

\bibitem[Krizhevsky and Hinton(2009)]{krizhevsky2009learning}
A.~Krizhevsky and G.~Hinton.
\newblock Learning multiple layers of features from tiny images.
\newblock 2009.

\bibitem[Lamport et~al.(1982)Lamport, Shostak, and Pease]{Lamport1982TheBG}
L.~Lamport, R.~E. Shostak, and M.~C. Pease.
\newblock The byzantine generals problem.
\newblock \emph{ACM Trans. Program. Lang. Syst.}, 4:\penalty0 382--401, 1982.

\bibitem[Lee et~al.(2017)Lee, Hwang, Park, and Kim]{lee2017risk}
J.~Lee, D.~Hwang, J.~Park, and K.-H. Kim.
\newblock Risk analysis and countermeasure for bit-flipping attack in lorawan.
\newblock In \emph{Information Networking (ICOIN), 2017 International
  Conference on}, pages 549--551. IEEE, 2017.

\bibitem[Li et~al.(2014{\natexlab{a}})Li, Andersen, Park, Smola, Ahmed,
  Josifovski, Long, Shekita, and Su]{Li2014ScalingDM}
M.~Li, D.~G. Andersen, J.~W. Park, A.~J. Smola, A.~Ahmed, V.~Josifovski,
  J.~Long, E.~J. Shekita, and B.-Y. Su.
\newblock Scaling distributed machine learning with the parameter server.
\newblock In \emph{OSDI}, 2014{\natexlab{a}}.

\bibitem[Li et~al.(2014{\natexlab{b}})Li, Andersen, Smola, and
  Yu]{Li2014CommunicationED}
M.~Li, D.~G. Andersen, A.~J. Smola, and K.~Yu.
\newblock Communication efficient distributed machine learning with the
  parameter server.
\newblock In \emph{NIPS}, 2014{\natexlab{b}}.

\bibitem[Loosli et~al.(2007)Loosli, Canu, and Bottou]{loosli2007training}
G.~Loosli, S.~Canu, and L.~Bottou.
\newblock Training invariant support vector machines using selective sampling.
\newblock \emph{Large scale kernel machines}, pages 301--320, 2007.

\bibitem[McMahan et~al.(2017)McMahan, Moore, Ramage, Hampson, and
  y~Arcas]{McMahan2017CommunicationEfficientLO}
H.~B. McMahan, E.~Moore, D.~Ramage, S.~Hampson, and B.~A. y~Arcas.
\newblock Communication-efficient learning of deep networks from decentralized
  data.
\newblock In \emph{AISTATS}, 2017.

\bibitem[Mukkamala and Hein(2017)]{Mukkamala2017VariantsOR}
M.~C. Mukkamala and M.~Hein.
\newblock Variants of rmsprop and adagrad with logarithmic regret bounds.
\newblock In \emph{ICML}, 2017.

\bibitem[Yin et~al.(2018)Yin, Chen, Ramchandran, and
  Bartlett]{yin2018byzantine}
D.~Yin, Y.~Chen, K.~Ramchandran, and P.~Bartlett.
\newblock Byzantine-robust distributed learning: Towards optimal statistical
  rates.
\newblock \emph{arXiv preprint arXiv:1803.01498}, 2018.

\end{thebibliography}
\bibliographystyle{abbrvnat}

\newpage
\clearpage
\section{Appendix}
In the appendix, we introduce several useful lemmas and use them to derive the detailed proofs of the theorems in this paper.

\subsection{Dimensional Byzantine resilience}

\setcounter{lemma}{0}
\begin{lemma}[\citet{blanchard2017machine}]
Let $v_1, \ldots, v_m$ be any i.i.d. random $d$-dimensional vectors s.t. $v_i \sim G$, with $\E[G] = g$ and $\E\|G-g\|^2 \leq V$. $q$ of $\{\tilde{v}_i: i \in [m]\}$ are Byzantine. If $2q + 2 < m$, we have 
$
\E \|\krum(\{\tilde{v}_i: i \in [m]\}) - g \|^2 \leq \Delta_0,
$
where 
$
\Delta_0 = \left(6m-6q + \frac{4q(m-q-2) + 4q^2(m-q-1)}{m-2q-2} \right)V.
$
\end{lemma}
\begin{proof}
We denote the $(m-q)$ correct values as $\{v_1, \ldots, v_{m-q}\}$, and $Kr = \krum(\{\tilde{v}_i: i \in [m]\})$.
Using \citet{blanchard2017machine} Proposition 1, we have 
\begin{align*}
\E \left\|Kr - \frac{\sum_{Kr \rightarrow \mbox{correct } j} \tilde{v}_j}{|Kr \rightarrow \mbox{correct } j|}  \right\|^2 \leq 2\left(m-q + \frac{q(m-q-2) + q^2(m-q-1)}{m-2q-2} \right)V,
\end{align*}
where $Kr \rightarrow \mbox{correct } j$ is the set of correct elements in the $m-q-2$ nearest neighbours to $Kr$ in $\{\tilde{v}_i: i \in [m]\}$, measured by Euclidean distance.
Thus, we obtain 
\begin{align*}
& \E \|\krum(\{\tilde{v}_i: i \in [m]\}) - g \|^2 \\
&\leq 2\E \left\|Kr - \frac{\sum_{Kr \rightarrow \mbox{correct } j} \tilde{v}_j}{|Kr \rightarrow \mbox{correct } j|}  \right\|^2 + 2\E \left\|g - \frac{\sum_{Kr \rightarrow \mbox{correct } j} \tilde{v}_j}{|Kr \rightarrow \mbox{correct } j|}  \right\|^2 \\
&\leq 4\left(m-q + \frac{q(m-q-2) + q^2(m-q-1)}{m-2q-2} \right)V + 2 \frac{\sum_{Kr \rightarrow \mbox{correct } j} \E \left\|g -\tilde{v}_j  \right\|^2}{|Kr \rightarrow \mbox{correct } j|}  \\
&\leq 4\left(m-q + \frac{q(m-q-2) + q^2(m-q-1)}{m-2q-2} \right)V + 2 (m-q) V  \\
&= \left(6m-6q + \frac{4q(m-q-2) + 4q^2(m-q-1)}{m-2q-2} \right)V.
\end{align*}
\end{proof}

\setcounter{theorem}{0}
\setcounter{proposition}{0}
\begin{proposition}
Averaging is not dimensional Byzantine resilient.
\end{proposition}
\begin{proof}
We demonstrate a counter example.
Consider the case where
\begin{align}
\tilde{v}_i = 
\begin{cases}
v_i, &\forall i \in [m-1]\\
-g - \sum_{i=1}^{m-1} v_i, &i=m,
\end{cases}
\end{align}
where $g = \E[v_i]$, $\forall i \in [m]$. Thus, the resulting aggregation is $Aggr = -g/m$.
The inner product $\ip{\E[Aggr]}{g}$ is always negative under the Byzantine attack. Thus, SGD is not expectedly descendant, which means it will not converge to critical points. Note that in this counter example, the number of Byzantine values of each dimension is $1$.

Hence,  averaging is not dimensional Byzantine-resilient.
\end{proof}

\begin{proposition}
Any aggregation rule $\aggr(\{\tilde{v}_i: i \in [m]\})$ that outputs $Aggr \in \{\tilde{v}_i: i \in [m]\}$ is not dimensional Byzantine resilient.
\end{proposition}
\begin{proof}
We demonstrate a counter example.
Consider the case where the $i$th dimension of the $i$th vector $v_i$ is manipulated by the malicious workers (e.g. multiplied by an arbitrarily large negative value), where $i \in [m]$. Thus, up to 1 value of each dimension is Byzantine. However, no matter which vector is chosen, as long as the aggregation is chosen from $\{\tilde{v}_i: i \in [m]\}$, the inner product $\ip{\E[Aggr]}{g}$ can be arbitrarily large negative value under the Byzantine attack. Thus, SGD is not expectedly descendant, which means it will not converge to critical points.

Hence,  any aggregation rule that outputs $Aggr \in \{\tilde{v}_i: i \in [m]\}$ is not dimensional Byzantine-resilient.
\end{proof}

\subsection{Trimmed mean}
We use the following lemma to bound the one-dimensional trimmed mean.

\begin{lemma}
Assume that among the scalar sequence $\{\tilde{v}_i: i \in [m]\}$, $q$ elements are Byzantine. Without loss of generality, we denote the remaining correct values as $\{v_1, \ldots, v_{m-q}\}$. Thus, for $q < b \leq \lceil m/2 \rceil - 1$, $v_{(b-q+i):(m-q)} \leq \tilde{v}_{(b+i):m} \leq v_{(b+i):(m-q)}$, for $\forall i \in [m-2b]$, where $\tilde{v}_{(b+i):m}$ is the $(b+i)$th smallest element in $\{\tilde{v}_i: i \in [m]\}$, and $v_{(b+i):(m-q)}$ is the $(b+i)$th smallest element in $\{v_1, \ldots, v_{m-q}\}$.
\label{lem:bound_ord_stat}
\end{lemma}

\begin{proof}
We prove the two inequalities separately. \\
(i) We prove the first inequality $v_{(b-q+i):(m-q)} \leq \tilde{v}_{(b+i):m}$ by contradiction. \\
If $v_{(b-q+i):(m-q)} > \tilde{v}_{(b+i):m}$, then there will be $\left( (m-q) - (b-q+i) + 1 \right) = (m-b-i+1)$ correct values larger than $\tilde{v}_{(b+i):m}$. However, because $\tilde{v}_{(b+i):m}$ is the $(b+i)$-th smallest element in the sequence 
$\{\tilde{v}_i: i \in [m]\}$, there is at most $\left( m - (b+i) \right) = (m - b - i)$ elements larger than $\tilde{v}_{(b+i):m}$, which yields a contradiction. 

(ii) We prove the second inequality $\tilde{v}_{(b+i):m} \leq v_{(b+i):(m-q)}$ by contradiction. \\
If $\tilde{v}_{(b+i):m} > v_{(b+i):(m-q)}$, then there will be $(b+i)$ correct values smaller than $\tilde{v}_{(b+i):m}$. However, because $\tilde{v}_{(b+i):m}$ is the $(b+i)$-th smallest element in the sequence 
$\{\tilde{v}_i: i \in [m]\}$, there is at most $(b+i-1)$ elements smaller than $\tilde{v}_{(b+i):m}$, which yields a contradiction. 
\end{proof}

\begin{theorem}(Bounded Variance)
Let $v_1, \ldots, v_m$ be any i.i.d. random $d$-dimensional vectors s.t. $v_i \sim G$, with $\E[G] = g$ and $\E\|G-g\|^2 \leq V$. In each dimension, $q$ values are Byzantine, which yields $\{\tilde{v}_i: i \in [m]\}$. If $2q < m$, we have 
$
\E \|\trmean_b(\{\tilde{v}_i: i \in [m]\}) - g \|^2 \leq \Delta_1,
$
where 
$
\Delta_1 = \frac{2(b+1)(m-q)}{(m-b-q)^2} V.
$
\end{theorem}

\begin{proof}
We first assume that all the $v_i$'s, $\tilde{v}_i$'s, and $g$ are scalars, with the variance $V = \sigma^2$.
Using Lemma~\ref{lem:bound_ord_stat}, we have
\begin{align*}
\quad&\sum_{i=b-q+1}^{m-q-b} (v_{i:(m-q)}-g) \leq \sum_{i=b+1}^{m-b} (\tilde{v}_{i:m}-g) \leq \sum_{i=b+1}^{m-b} (v_{i:(m-q)}-g) \\
\Rightarrow & \frac{\sum_{i=1}^{m-q-b} (v_{i:(m-q)}-g)}{m-b-q} \leq \frac{\sum_{i=b+1}^{m-b} (\tilde{v}_{i:m}-g)}{m-2b} \leq \frac{\sum_{i=b+1}^{m-q} (v_{i:(m-q)}-g)}{m-b-q} \\
\Rightarrow& \left[ \frac{\sum_{i=b+1}^{m-b} (\tilde{v}_{i:m}-g)}{m-2b} \right]^2 \leq \max \left\{\left[ \frac{\sum_{i=1}^{m-q-b} (v_{i:(m-q)}-g)}{m-b-q} \right]^2, \left[ \frac{\sum_{i=b+1}^{m-q} (v_{i:(m-q)}-g)}{m-b-q} \right]^2 \right\}.
\end{align*}
Thus, we have 
\begin{align*}
&\left[ \trmean_b(\{\tilde{v}_i: i \in [m]\}) - g \right]^2 \\
&= \left[ \frac{\sum_{i=b+1}^{m-b} \tilde{v}_{i:m}}{m-2b} - g \right]^2 \\
&\leq \max \left\{\left[ \frac{\sum_{i=1}^{m-q-b} (v_{i:(m-q)}-g)}{m-b-q} \right]^2, \left[ \frac{\sum_{i=b+1}^{m-q} (v_{i:(m-q)}-g)}{m-b-q} \right]^2 \right\}.
\end{align*}
Note that for arbitrary subset $\mathcal{S} \subseteq [m-q]$, $|\mathcal{S}| = m-b-q$, we have the following bound:
\begin{align*}
&\left[ \frac{\sum_{i \in \mathcal{S}} (v_{i:(m-q)}-g)}{m-b-q } \right]^2 \\
&= \left[ \frac{\sum_{i \in [m-q]} (v_{i:(m-q)}-g) - \sum_{i \notin \mathcal{S}} (v_{i:(m-q)}-g)}{m-b-q} \right]^2 \\
&\leq 2\left[ \frac{\sum_{i \in [m-q]} (v_{i:(m-q)}-g)}{m-b-q} \right]^2 + 2\left[ \frac{\sum_{i \notin \mathcal{S}} (v_{i:(m-q)}-g)}{m-b-q} \right]^2 \\
&= \frac{2(m-q)^2}{(m-b-q)^2} \left[ \frac{\sum_{i \in [m-q]} (v_{i:(m-q)}-g)}{m-q} \right]^2 + \frac{2b^2}{(m-b-q)^2} \left[ \frac{\sum_{i \notin \mathcal{S}} (v_{i:(m-q)}-g)}{b} \right]^2 \\
&\leq \frac{2(m-q)^2}{(m-b-q)^2} \left[ \frac{\sum_{i \in [m-q]} (v_{i:(m-q)}-g)}{m-q} \right]^2 
+ \frac{2b^2}{(m-b-q)^2} \frac{\sum_{i \notin \mathcal{S}} (v_{i:(m-q)}-g)^2}{b}  \\
&\leq \frac{2(m-q)^2}{(m-b-q)^2} \left[ \frac{\sum_{i \in [m-q]} (v_{i:(m-q)}-g)}{m-q} \right]^2 
+ \frac{2b^2}{(m-b-q)^2} \frac{\sum_{i \in [m-q]} (v_{i:(m-q)}-g)^2}{b}.
\end{align*}
By taking the expectations, we obtain 
\begin{align*}
&\E \left[ \frac{\sum_{i \in \mathcal{S}} (v_{i:(m-q)}-g)}{m-b-q} \right]^2 \\
&\leq \frac{2(m-q)^2}{(m-b-q)^2} \frac{\sigma^2}{m-q} 
+ \frac{2b^2}{(m-b-q)^2} \frac{(m-q)\sigma^2}{b} \\
&= \frac{2(m-q) \sigma^2}{(m-b-q)^2} 
+ \frac{2b(m-q)\sigma^2}{(m-b-q)^2} \\
&= \frac{2(b+1)(m-q) \sigma^2}{(m-b-q)^2}. 
\end{align*}
Combining all the ingredients above, we obtain the desired result:
\begin{align*}
\E \left[ \trmean_b(\{\tilde{v}_i: i \in [m]\}) - g \right]^2 \leq \frac{2(b+1)(m-q) \sigma^2}{(m-b-q)^2}.
\end{align*}
Then, we generalize  $v_i$'s, $\tilde{v}_i$'s, and $g$ to $d$-dimensional vectors with the variance $V = \sum_{j=1}^d \sigma_j^2$, where $\E\|G_j - g_j\|^2 \leq \sigma_j^2$, $\forall j \in [d]$. For $\forall j \in [d]$, we have
\begin{align*}
\E \left[ \trmean_b(\{(\tilde{v}_i)_j: i \in [m]\}) - g_j \right]^2 \leq \frac{2(b+1)(m-q) \sigma_j^2}{(m-b-q)^2}.
\end{align*}
Thus, we have
\begin{align*}
\E \left[ \trmean_b(\{\tilde{v}_i: i \in [m]\}) - g \right]^2 \leq \frac{2(b+1)(m-q) \sum_{j=1}^d \sigma_j^2}{(m-b-q)^2} = \frac{2(b+1)(m-q)}{(m-b-q)^2} V.
\end{align*}
\end{proof}

\subsection{Phocas}

The following lemma bounds the one-dimensional $Phocas_b(\cdot)$.
\begin{lemma}
For the scalar sequence $\{\tilde{v}_i \in [m]\}$, and the corresponding trimmed mean $Trmean_b = \trmean_b(\{\tilde{v}_i \in [m]\})$, we have 
\begin{align*}
|\tilde{v}_{i/Trmean_b} - Trmean_b| \leq |v_{i/Trmean_b} - Trmean_b|,
\end{align*}
where $i \in [m-q]$, $\{v_i: i \in [m-q]\}$ is the sequence of $(m-q)$ correct values in $\{\tilde{v}_i \in [m]\}$.
\end{lemma}
\begin{proof}
We prove this lemma by contradiction.

Assume that $|\tilde{v}_{i/Trmean_b} - Trmean_b| > |v_{i/Trmean_b} - Trmean_b|$. Thus, there are $i$ correct values closer than $\tilde{v}_{i/Trmean_b}$ to $Trmean_b$. However, according to the definition of $\tilde{v}_{i/Trmean_b}$, it is the $i$th closest value to $Trmean_b$, which means that there are at most $(i-1)$ values closer than $\tilde{v}_{i/Trmean_b}$ to $Trmean_b$, which yields a contradiction.
\end{proof}

\begin{theorem}(Bounded Variance)
Let $v_1, \ldots, v_n$ be any i.i.d. random $d$-dimensional vectors s.t. $v_i \sim G$, with $\E[G] = g$ and $\E\|G-g\|^2 \leq V$. In each dimension, $q$ values are Byzantine, which yields $\{\tilde{v}_i: i \in [m]\}$ If $2q < n$, we have 
$
\E \|\phocas_b(\{\tilde{v}_i: i \in [m]\}) - g \|^2 \leq \Delta_2,
$
where 
$
\Delta_2 = \left[4 + \frac{12(b+1)(m-q)}{(m-b-q)^2} \right] V.
$
\end{theorem}
\begin{proof}
We first assume that all the $v_i$'s, $\tilde{v}_i$'s, and $g$ are scalars, with the variance $V = \sigma^2$. 
For convenience, we denote $Trmean_b = \trmean_b(\{\tilde{v}_i \in [m]\})$. 
Thus, we have
\begin{align*}
& \left[ \phocas_b(\{\tilde{v}_i: i \in [m]\}) - g \right]^2 \\
&\leq 2\left[ \phocas_b(\{\tilde{v}_i: i \in [m]\}) - Trmean_b \right]^2 + 2\left[ g - Trmean_b \right]^2.
\end{align*}
Using Theorem~\ref{thm:bound_trmean_var}, we already have 
\begin{align*}
\E\left[ g - Trmean_b \right]^2 \leq  \frac{2(b+1)(m-q) \sigma^2}{(m-b-q)^2}.
\end{align*}
We only need to bound $\left[ \phocas_b(\{\tilde{v}_i: i \in [m]\}) - Trmean_b \right]^2$ as follows:
\begin{align*}
& \left[ \phocas_b(\{\tilde{v}_i: i \in [m]\}) - Trmean_b \right]^2 \\
&= \left[ \frac{\sum_{i=1}^{m-b} \tilde{v}_{i/Trmean_b} - Trmean_b}{m-b} \right]^2 \\
&\leq \frac{\sum_{i=1}^{m-b} \left[ \tilde{v}_{i/Trmean_b} - Trmean_b \right]^2}{m-b} \\
&\leq \frac{\sum_{i=1}^{m-b} \left[ v_{i/Trmean_b} - Trmean_b \right]^2}{m-b} \\
&\leq \frac{\sum_{i=1}^{m-q} \left[ v_{i/Trmean_b} - Trmean_b \right]^2}{m-q} \\
&= \frac{\sum_{i=1}^{m-q} \left[ v_i - Trmean_b \right]^2}{m-q} \\
&\leq \frac{\sum_{i=1}^{m-q} 2\left[ v_i - g \right]^2}{m-q} + \frac{\sum_{i=1}^{m-q} 2\left[ g - Trmean_b \right]^2}{m-q} \\
&\leq \frac{\sum_{i=1}^{m-q} 2\left[ v_i - g \right]^2}{m-q} + 2\left[ g - Trmean_b \right]^2.
\end{align*}
Taking expectation on both sides, we have 
\begin{align*}
\E\left[ \phocas_b(\{\tilde{v}_i: i \in [m]\}) - Trmean_b \right]^2 
\leq 2\sigma^2 + 2\E\left[ g - Trmean_b \right]^2.
\end{align*}
Thus, we have
\begin{align*}
& \E\left[ \phocas_b(\{\tilde{v}_i: i \in [m]\}) - g \right]^2 \\
&\leq 2\E\left[ Phocas_b(\{\tilde{v}_i: i \in [m]\}) - Trmean_b \right]^2 + 2\E\left[ g - Trmean_b \right]^2 \\
&\leq 4\sigma^2 + 6\E\left[ g - Trmean_b \right]^2 \\
& \leq 4\sigma^2 + \frac{12(b+1)(m-q) \sigma^2}{(m-b-q)^2}.
\end{align*}
Then, generalizing  $v_i$'s, $\tilde{v}_i$'s, and $g$ to $d$-dimensional vectors with the variance $V$, we obtain the desired result
\begin{align*}
\E\left[ \phocas_b(\{\tilde{v}_i: i \in [m]\}) - g \right]^2 \leq \left[4 + \frac{12(b+1)(m-q)}{(m-b-q)^2} \right] V.
\end{align*}
\end{proof}

\subsection{Convergence analysis}

\subsubsection{Strongly convex functions}

The following lemma is from \cite{bubeck2015convex}.
\begin{lemma}
\label{lem:bubeck_cvx}
Let $F$ be $\mu_F$-strongly convex and $L_F$-smooth. Then for all $x, y \in \R^d$, one has
\begin{align*}
\ip{\nabla F(x) - \nabla F(y)}{x-y} \geq \frac{\mu_F L_F}{\mu_F + L_F} \|x-y\|^2 + \frac{1}{\mu_F + L_F} \|\nabla F(x) - \nabla F(y)\|^2.
\end{align*}
\end{lemma}

\begin{theorem}
Assume that $F(x)$ is $\mu_F$-strongly convex and $L_F$-smooth, where $0 < \mu_F \leq L_F$. We take $\gamma \leq \frac{2}{\mu_F + L_F}$. In any iteration $t$, the correct gradients are $v_i^t = \nabla F_i(x^t)$. Using any (classic or dimensional) $\Delta$-Byzantine-resilient aggregation rule with corresponding assumptions, we obtain linear convergence with a constant error after $T$ iterations with synchronous SGD:
\begin{align*}
\E\|x^{T} - x^*\| \leq \left( 1- \frac{\gamma \mu_F L_F}{\mu_F + L_F} \right)^T \|x^0-x^*\| + \frac{\mu_F + L_F}{\mu_F L_F} \gamma  \sqrt{\Delta}.
\end{align*}
\end{theorem}
\begin{proof}
Denote the aggregation rule as $\aggr(\cdot)$, and $g(x^t) = \aggr(\{\tilde{v}_i^t: i \in [m]\})$, and $x^{t+1} = x^t - \gamma g(x^t)$.

Thus, we have 
\begin{align}
\label{equ:str_cvx_step_decomp}
\|x^{t+1} - x^*\| = \| x^t - \gamma g(x^t) - x^* \| \leq \| x^t - \gamma \nabla F(x^t) - x^* \| + \| \gamma \nabla F(x^t) - \gamma g(x^t) \|.
\end{align}

Furthermore, we have 
\begin{align*}
\| x^t - \gamma \nabla F(x^t) - x^* \|^2 = \|x^t - x^*\|^2 + \gamma^2 \|\nabla F(x^t)\|^2 - 2\gamma \ip{x^t - x^*}{\nabla F(x^t)}.
\end{align*}

Using the $\mu_F$-strong convexity, $L_F$-smoothness of $F(x)$, and Lemma~\ref{lem:bubeck_cvx} with $x = x^t$, $y = x^*$, we have
\begin{align*}
\ip{\nabla F(x^t)}{x^t-x^*} \geq \frac{\mu_F L_F}{\mu_F + L_F} \|x^t-x^*\|^2 + \frac{1}{\mu_F + L_F} \|\nabla F(x^t)\|^2.
\end{align*}
Taking $\gamma \leq \frac{2}{\mu_F + L_F}$, we obtain
\begin{align*}
&\| x^t - \gamma \nabla F(x^t) - x^* \|^2 \\
&\leq \|x^t - x^*\|^2 + \gamma^2 \|\nabla F(x^t)\|^2 - 2\gamma \left( \frac{\mu_F L_F}{\mu_F + L_F} \|x^t-x^*\|^2 + \frac{1}{\mu_F + L_F} \|\nabla F(x^t)\|^2 \right) \\
&\leq \left( 1- \frac{2\gamma \mu_F L_F}{\mu_F + L_F} \right) \|x^t-x^*\|^2 + \gamma \left( \gamma - \frac{2}{\mu_F + L_F} \right) \|\nabla F(x^t)\|^2 \\
&\leq \left( 1- \frac{2\gamma \mu_F L_F}{\mu_F + L_F} \right) \|x^t-x^*\|^2.
\end{align*}

Note that when $\gamma \leq \frac{2}{\mu_F + L_F}$, we have $\frac{2\gamma \mu_F L_F}{\mu_F + L_F} \leq \frac{4 \mu_F L_F}{(\mu_F + L_F)^2} \leq 1$. Using $\sqrt{1-a} \leq  1-\frac{a}{2}$, $\forall a \leq 1$, we obtain 
\begin{align*}
\| x^t - \gamma \nabla F(x^t) - x^* \| \leq \left( 1- \frac{\gamma \mu_F L_F}{\mu_F + L_F} \right) \|x^t-x^*\|.
\end{align*}

Combined with Equation~\ref{equ:str_cvx_step_decomp}, we obtain
\begin{align*}
\|x^{t+1} - x^*\| \leq \left( 1- \frac{\gamma \mu_F L_F}{\mu_F + L_F} \right) \|x^t-x^*\| + \gamma \| \nabla F(x^t) - g(x^t) \|.
\end{align*}

Conditional on $x^t$, taking the expectation on both sides, we obtain 
\begin{align*}
&\E\|x^{t+1} - x^*\| \\
&\leq \left( 1- \frac{\gamma \mu_F L_F}{\mu_F + L_F} \right) \|x^t-x^*\| + \gamma \E\| \nabla F(x^t) - g(x^t) \| \\
&\leq \left( 1- \frac{\gamma \mu_F L_F}{\mu_F + L_F} \right) \|x^t-x^*\| + \gamma \sqrt{\E\| \nabla F(x^t) - g(x^t) \|^2} \comment{Jenssen's inequality} \\
&\leq \left( 1- \frac{\gamma \mu_F L_F}{\mu_F + L_F} \right) \|x^t-x^*\| + \gamma \sqrt{\Delta}. \comment{$\Delta$-Byzantine resilience}
\end{align*}

By telescoping and taking total expectation, we obtain
\begin{align*}
&\E\|x^{T} - x^*\| \\
&\leq \left( 1- \frac{\gamma \mu_F L_F}{\mu_F + L_F} \right)^T \|x^0-x^*\| + \sum_{i=0}^{T-1} \left( 1- \frac{\gamma \mu_F L_F}{\mu_F + L_F} \right)^i \gamma \sqrt{\Delta} \\
&\leq \left( 1- \frac{\gamma \mu_F L_F}{\mu_F + L_F} \right)^T \|x^0-x^*\| + \sum_{i=0}^{+\infty} \left( 1- \frac{\gamma \mu_F L_F}{\mu_F + L_F} \right)^i \gamma \sqrt{\Delta} \\
&\leq \left( 1- \frac{\gamma \mu_F L_F}{\mu_F + L_F} \right)^T \|x^0-x^*\| + \frac{\mu_F + L_F}{\mu_F L_F} \gamma \sqrt{\Delta}. \comment{$\sum_{i=0}^{+\infty} a^i = \frac{1}{1-a}$, for $|a| < 1$}
\end{align*}
\end{proof}

\subsubsection{General functions}
For general non-strongly convex functions and non-convex functions, we provide the following convergence guarantee.
\begin{theorem}
Assume that $F(x)$ is $L_F$-smooth and potentially non-convex, where $0 < L_F$. We take $\gamma \leq \frac{1}{L_F}$. In any iteration $t$, the correct gradients are $v_i^t = \nabla F_i(x^t)$. Using any (classic or dimensional) $\Delta$-Byzantine-resilient aggregation rule with corresponding assumptions, we obtain linear convergence with a constant error after $T$ iterations with synchronous SGD:
\begin{align*}
\frac{\sum_{i=0}^{T-1} \E\|\nabla F(x^i)\|^2}{T} \leq \frac{2}{\gamma T} \left[ F(x^0) - F(x^*) \right] + \Delta.
\end{align*}
\end{theorem}
\begin{proof}
Denote the aggregation rule as $\aggr(\cdot)$, and $g(x^t) = \aggr(\{\tilde{v}_i^t: i \in [m]\})$, and $x^{t+1} = x^t - \gamma g(x^t)$.

To prove the convergence with constant error, we first bound the descendant of the loss value in each iteration.
\begin{align*}
&F(x^{t+1}) \leq F(x^t) + \ip{\nabla F(x^t)}{x^{t+1} - x^t} + \frac{L_F}{2} \|x^{t+1} - x^t\|^2 \comment{smoothness} \\
&= F(x^t) - \gamma \ip{\nabla F(x^t)}{g(x^t)} + \frac{L_F \gamma^2}{2} \|g(x^t)\|^2 \\
&\leq F(x^t) - \gamma \ip{\nabla F(x^t)}{g(x^t)} + \frac{\gamma}{2} \|g(x^t)\|^2 \comment{$\gamma \leq \frac{1}{L_F}$} \\
&= F(x^t) - \frac{\gamma}{2} \|\nabla F(x^t)\|^2 + \frac{\gamma}{2} \|\nabla F(x^t) - g(x^t)\|^2.
\end{align*}
Thus, we obtain
\begin{align*}
F(x^{t+1}) - F(x^*) \leq F(x^t) - F(x^*) - \frac{\gamma}{2} \|\nabla F(x^t)\|^2 + \frac{\gamma}{2} \|\nabla F(x^t) - g(x^t)\|^2.
\end{align*}
By telescoping and taking total expectation, we obtain
\begin{align*}
0 \leq \E\left[ F(x^{T}) - F(x^*) \right] \leq F(x^0) - F(x^*) - \frac{\gamma}{2} \sum_{i=0}^{T-1} \E\|\nabla F(x^i)\|^2 + \frac{\gamma}{2} T\Delta.
\end{align*}

By rearranging the terms, we obtain the desired result
\begin{align*}
\frac{\sum_{i=0}^{T-1} \E\|\nabla F(x^i)\|^2}{T} \leq \frac{2}{\gamma T} \left[ F(x^0) - F(x^*) \right] + \Delta.
\end{align*}
\end{proof}

\newpage
\clearpage
\subsection{Experimental details}
\label{sec:sup_exp}
In Table~\ref{tbl:mlp} and \ref{tbl:cnn}, we show the detailed network structures of the MLP and CNN used in our experiments.
\begin{table}[htb!]
\centering
\caption{MLP Summary}
\label{tbl:mlp}
\begin{tabular}{|l|l|l|}
\hline 
Layer (type) & Parameters & Previous Layer \\ \hline 
flatten(Flatten) & null & data \\ \hline
fc1(FullyConnected) & \#output=128 & flatten \\ \hline
relu1(Activation) & null & fc1 \\ \hline
fc2(FullyConnected) & \#output=128 & relu1 \\ \hline
relu2(Activation) & null & fc2 \\ \hline
fc3(FullyConnected) & \#output=10 & relu2 \\ \hline
softmax(SoftmaxOutput) & null & fc3 \\ \hline
\end{tabular} 
\end{table}
\begin{table}[htb!]
\centering
\caption{CNN Summary}
\label{tbl:cnn}
\begin{tabular}{|l|l|l|}
\hline 
Layer (type) & Parameters & Previous Layer \\ \hline 
conv1(Convolution)& channels=32, kernel\_size=3, padding=1 &data \\ \hline 
activation1(Activation)& null &conv1 \\ \hline 
conv2(Convolution)& channels=32, kernel\_size=3, padding=1 &activation1 \\ \hline 
activation2(Activation)& null &conv2 \\ \hline 
pooling1(Pooling)& pool\_size=2 &activation2 \\ \hline 
dropout1(Dropout)& probability=0.2 &pooling1 \\ \hline 
conv3(Convolution)& channels=64, kernel\_size=3, padding=1 &dropout1 \\ \hline 
activation2(Activation)& null &conv3 \\ \hline 
conv4(Convolution)& channels=64, kernel\_size=3, padding=1 &activation2 \\ \hline 
activation4(Activation)& null &conv4 \\ \hline 
pooling2(Pooling)& pool\_size=2 &activation4 \\ \hline 
dropout2(Dropout)& probability=0.2 &pooling2 \\ \hline 
flatten1(Flatten)& null &dropout2 \\ \hline 
fc1(FullyConnected)& \#output=512 &flatten1 \\ \hline 
activation5(Activation)& null &fc1 \\ \hline 
dropout3(Dropout)& probability=0.2 &activation5 \\ \hline 
fc2(FullyConnected)& \#output=512 &dropout3 \\ \hline 
activation6(Activation)& null &fc2 \\ \hline 
dropout4(Dropout)& probability=0.2 &activation6 \\ \hline 
fc3(FullyConnected)& \#output=10 &dropout4 \\ \hline 
softmax(SoftmaxOutput)& null &fc3 \\ \hline   
\end{tabular} 
\end{table}

\newpage
\clearpage
\subsection{Additional experiments}

In this section, we illustrate the additional empirical results. 

In Figure~\ref{fig:mnist_batchsize}, we illustrate the top-1 accuracy of MLP on MNIST when batch-size varies, without Byzantine failures. The learning rate is 
\begin{align*}
\gamma = \frac{0.1 \times batchsize}{32}.
\end{align*}
The results show that when there is no Byzantine failures, \textit{Phocas} performs just like \textit{Mean}. \textit{Trmean} has slightly slower convergence. \textit{Krum} is the slowest. The gap is narrowed when the batch size increases.

We illustrate all the experimental results of CNN on CIFAR10 additional to Section~\ref{sec:experiments}. For completeness, we also illustrate the experimental results of MLP on MNIST. The results are shown in Figure~\ref{fig:mnist_nobyz_appendix}-\ref{fig:cifar10_multiserver_appendix}. In general, all the experiments of CNN on CIFAR10 show similar results with the experiments of MLP on MNIST.

\begin{figure*}[htb!]
\centering
\subfigure[MLP on MNIST without Byzantine with different batch sizes]{\includegraphics[width=0.90\textwidth]{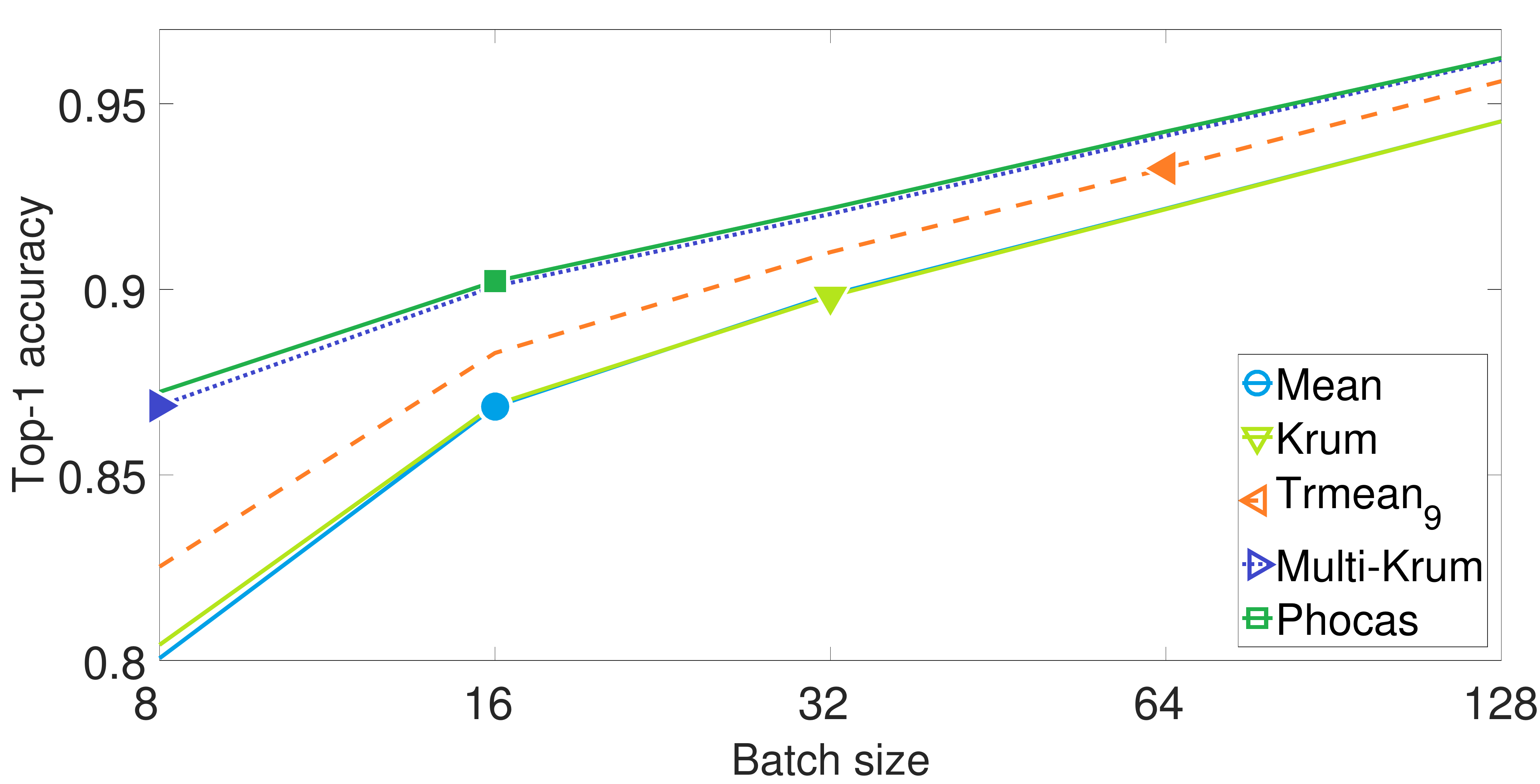}}
\caption{Top-1 accuracy of MLP on MNIST without Byzantine failures, when batch size varies. The learning rate is $\gamma = \frac{0.1 \times batchsize}{32}$.}
\label{fig:mnist_batchsize}
\end{figure*}

\begin{figure*}[htb!]
\centering
\subfigure[MLP on MNIST without Byzantine]{\includegraphics[width=0.90\textwidth]{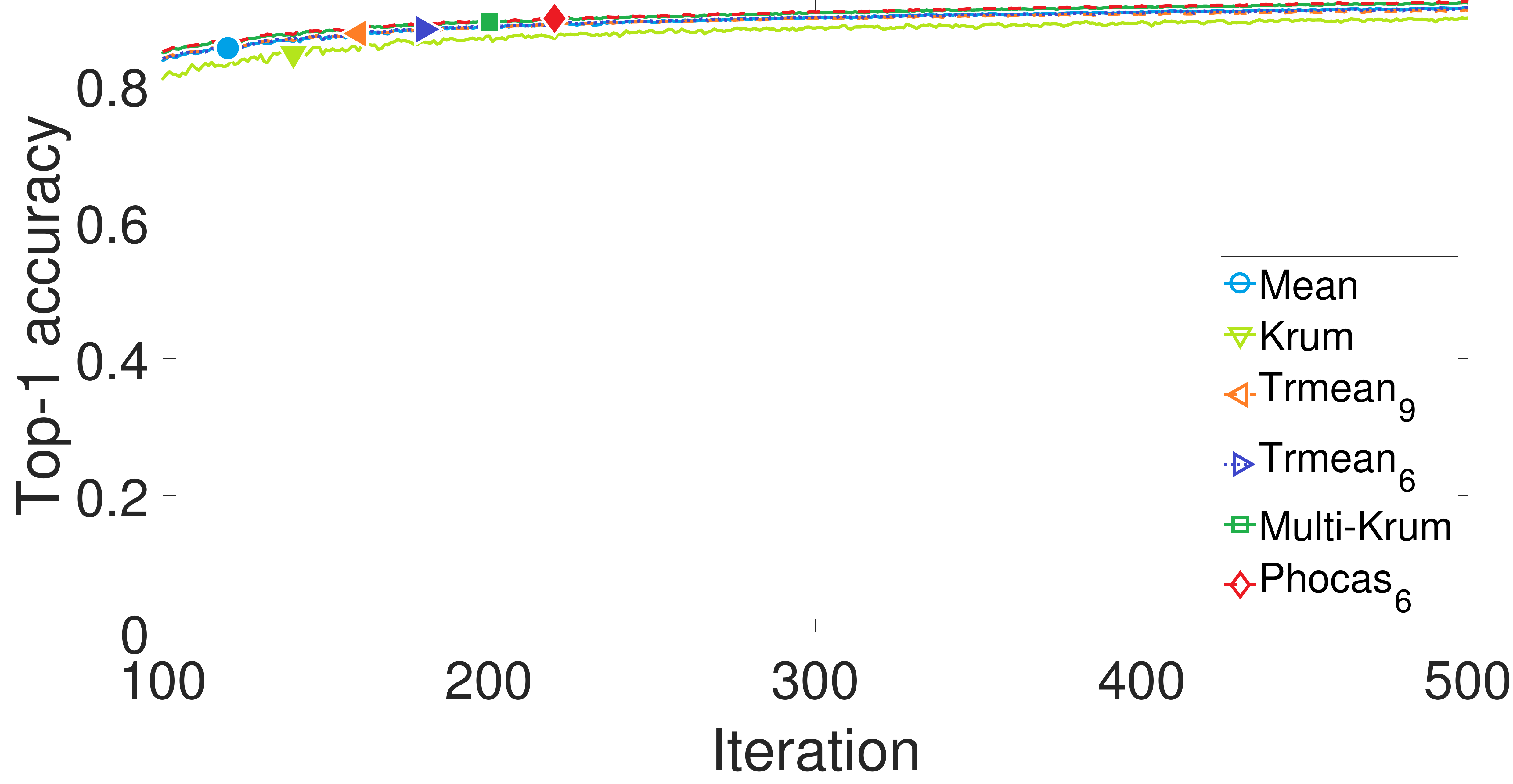}}
\subfigure[MLP on MNIST without Byzantine~(zoomed)]{\includegraphics[width=0.90\textwidth]{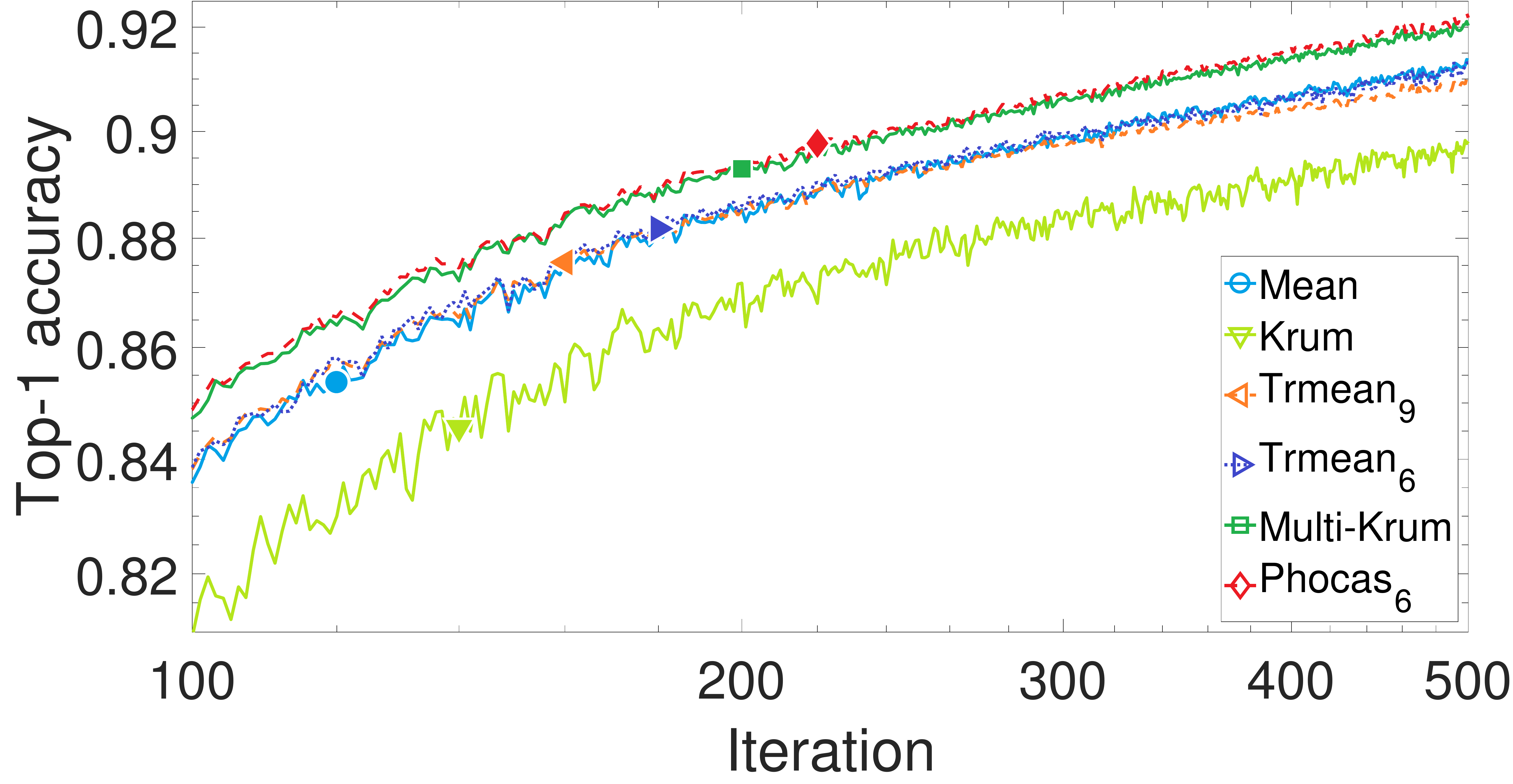}}
\caption{Top-1 accuracy of MLP on MNIST without Byzantine failures.}
\label{fig:mnist_nobyz_appendix}
\end{figure*}
\begin{figure*}[htb!]
\centering
\subfigure[MLP on MNIST with Gaussian]{\includegraphics[width=0.90\textwidth]{mnist_gaussian_small}}
\subfigure[MLP on MNIST with Gaussian~(zoomed)]{\includegraphics[width=0.90\textwidth]{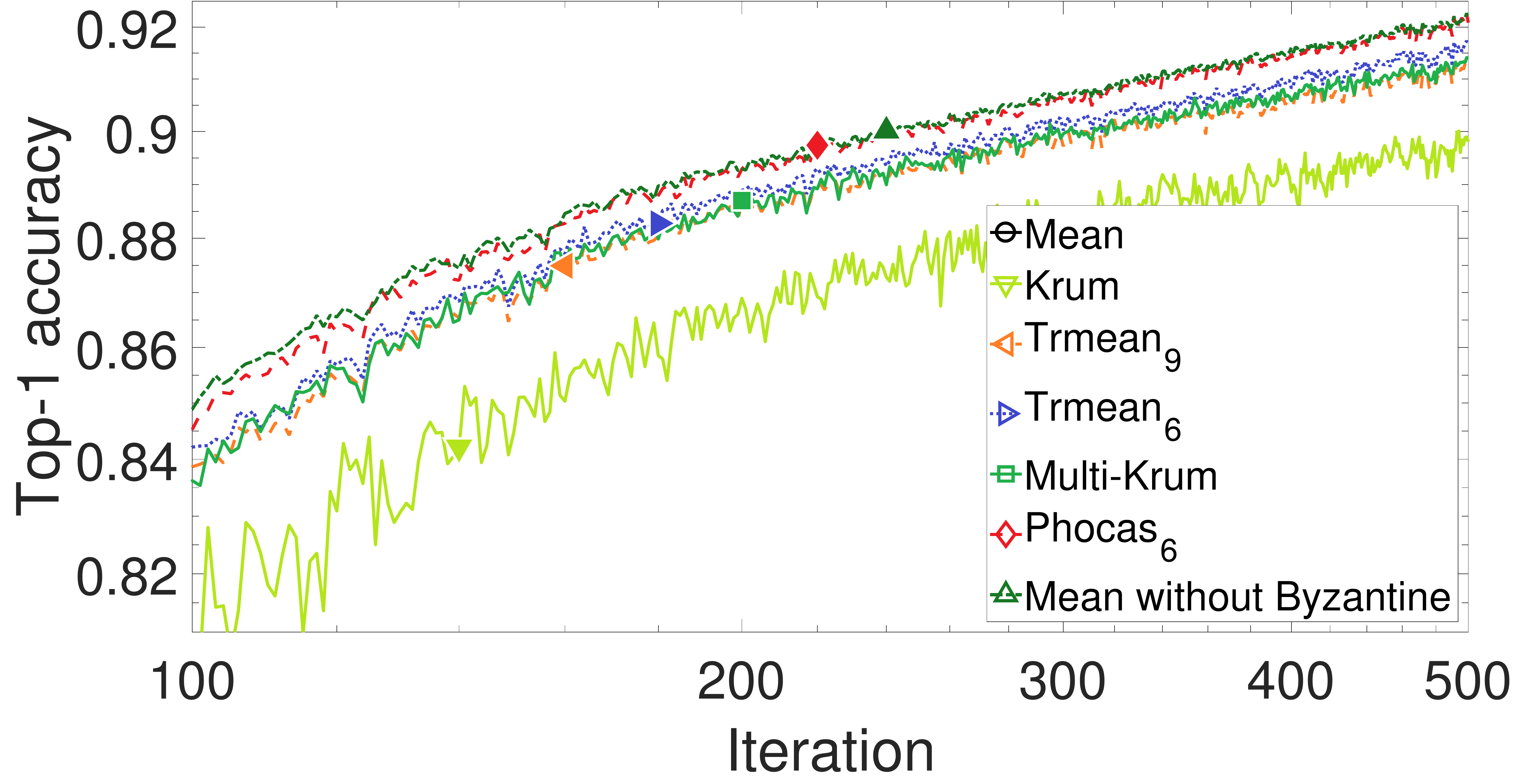}}
\caption{Top-1 accuracy of MLP on MNIST with Gaussian attack. 6 out of 20 gradient vectors are replaced by i.i.d. random vectors drawn from a Gaussian distribution with 0 mean and 200 standard deviation.}
\label{fig:mnist_gaussian_appendix}
\end{figure*}
\begin{figure*}[htb!]
\centering
\subfigure[MLP on MNIST with omniscient]{\includegraphics[width=0.90\textwidth]{mnist_omniscient_small}}
\subfigure[MLP on MNIST with omniscient~(zoomed)]{\includegraphics[width=0.90\textwidth]{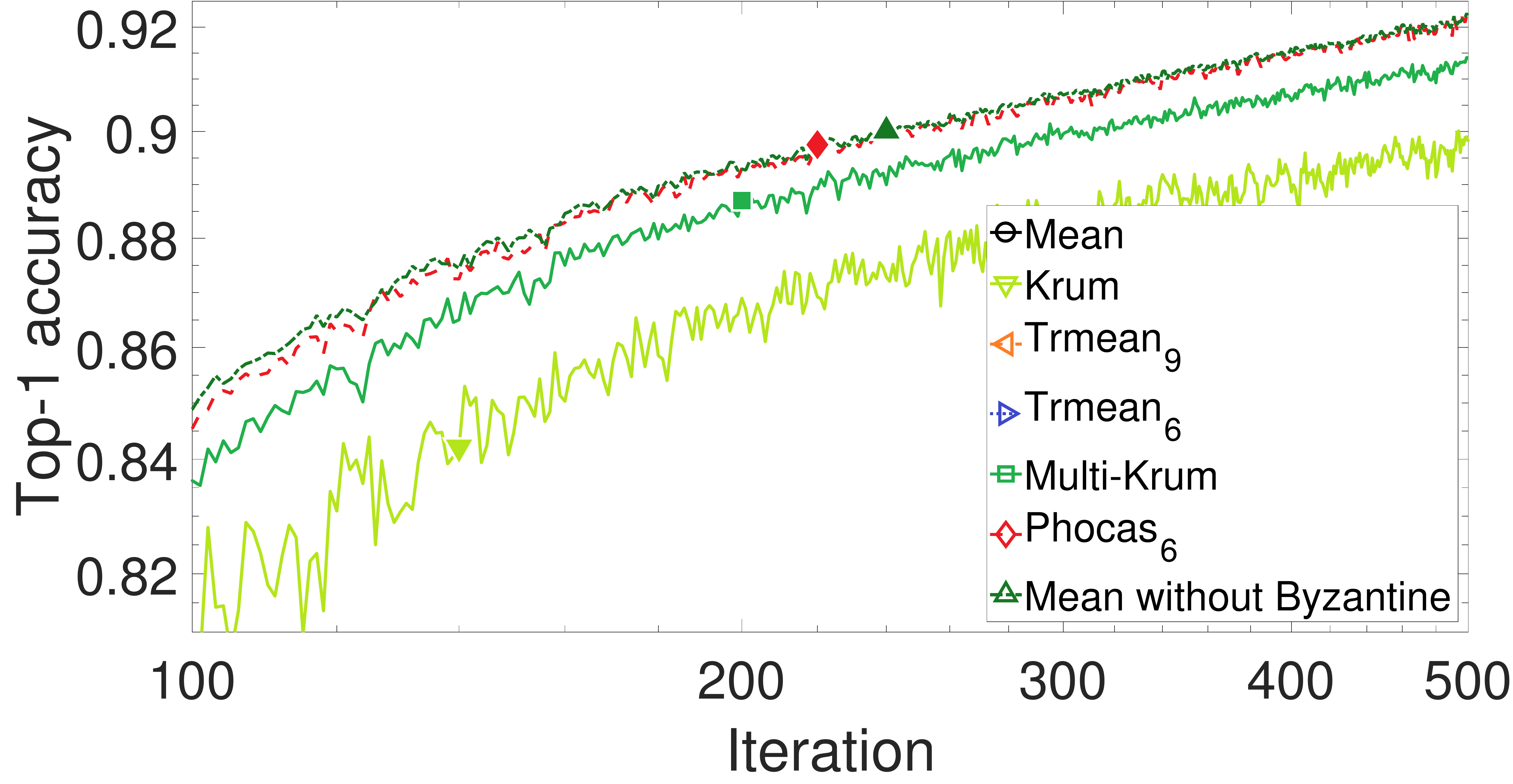}}
\caption{Top-1 accuracy of MLP on MNIST with omniscient attack. 6 out of 20 gradient vectors are replaced by the negative sum of all the correct gradients, scaled by a large constant~(1e20 in the experiments). }
\label{fig:mnist_omniscient_appendix}
\end{figure*}
\begin{figure*}[htb!]
\centering
\subfigure[MLP on MNIST with bit-flip]{\includegraphics[width=0.90\textwidth]{mnist_bitflip_small}}
\subfigure[MLP on MNIST with bit-flip~(zoomed)]{\includegraphics[width=0.90\textwidth]{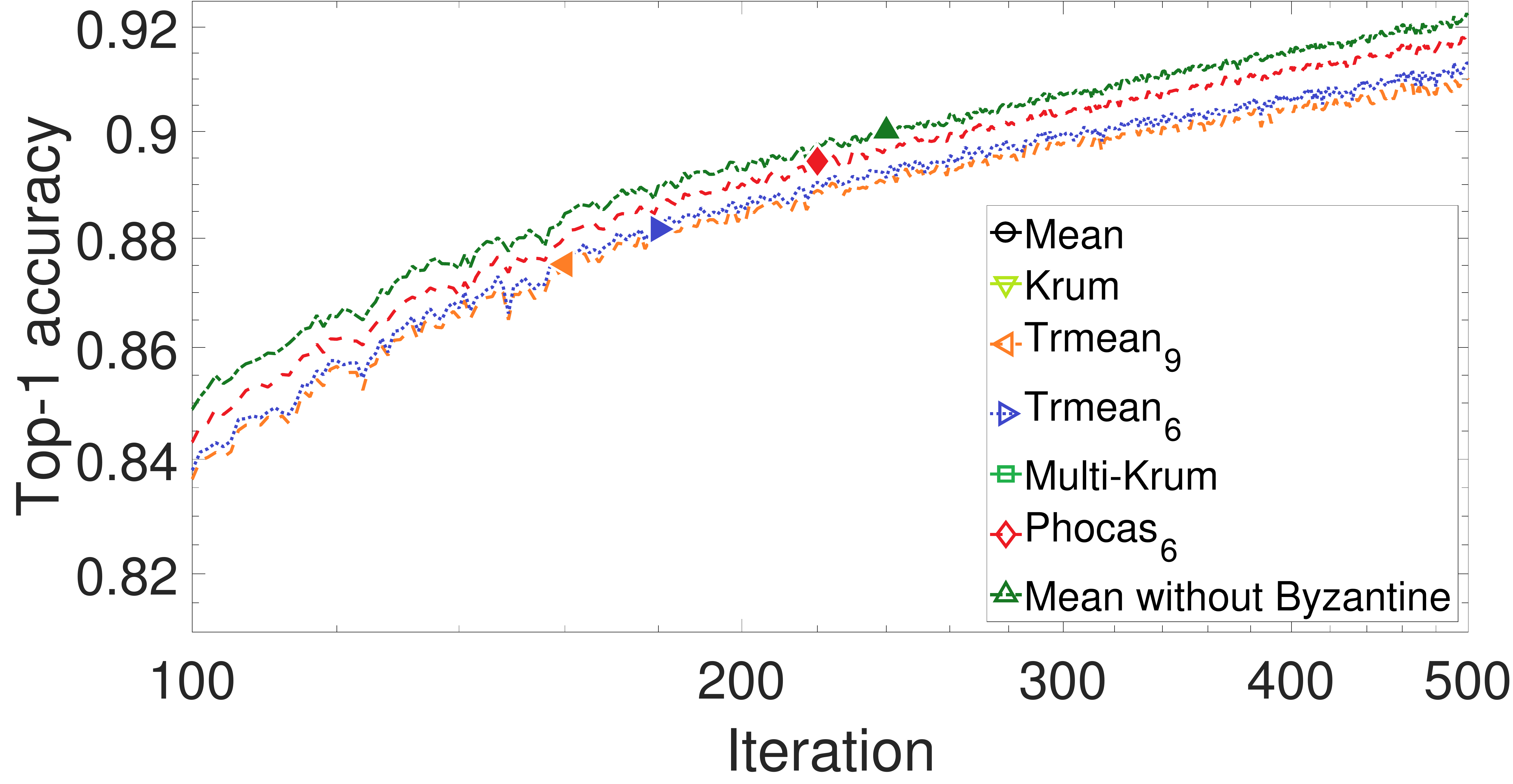}}
\caption{Top-1 accuracy of MLP on MNIST with bit-flip attack. For the first 1000 dimensions, 1 of the 20 floating numbers is manipulated by flipping the 22th, 30th, 31th and 32th bits.}
\label{fig:mnist_bitflip_appendix}
\end{figure*}
\begin{figure*}[htb!]
\centering
\subfigure[MLP on MNIST with gambler]{\includegraphics[width=0.90\textwidth]{mnist_multiserver_small}}
\subfigure[MLP on MNIST with gambler~(zoomed)]{\includegraphics[width=0.90\textwidth]{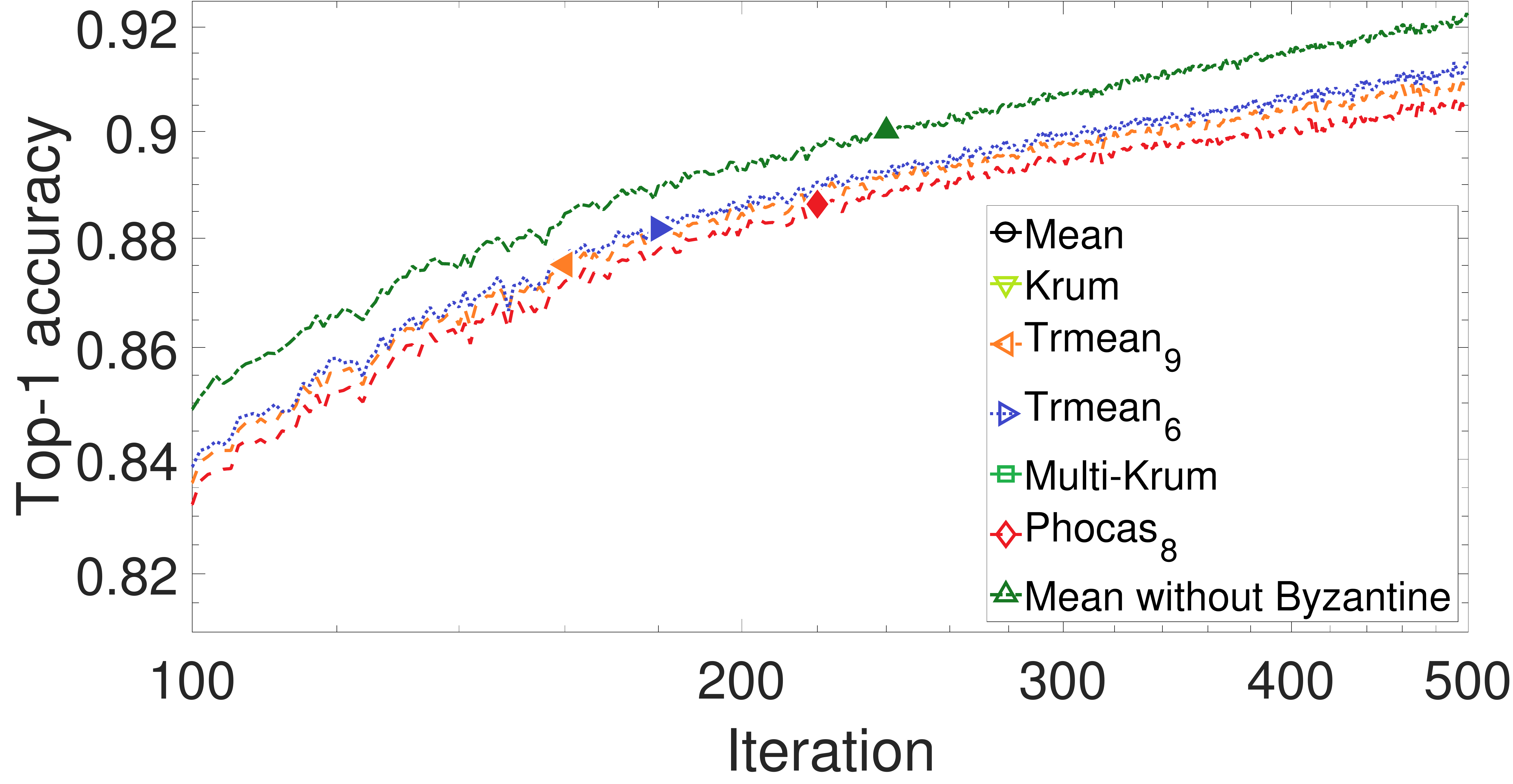}}
\caption{Top-1 accuracy of MLP on MNIST with gambler attack. The parameters are evenly assigned to 20 servers. For one single server, any received value is multiplied by $-1e20$ with probability 0.05\%.}
\label{fig:mnist_multiserver_appendix}
\end{figure*}

\begin{figure*}[htb]
\centering
\includegraphics[width=0.90\textwidth]{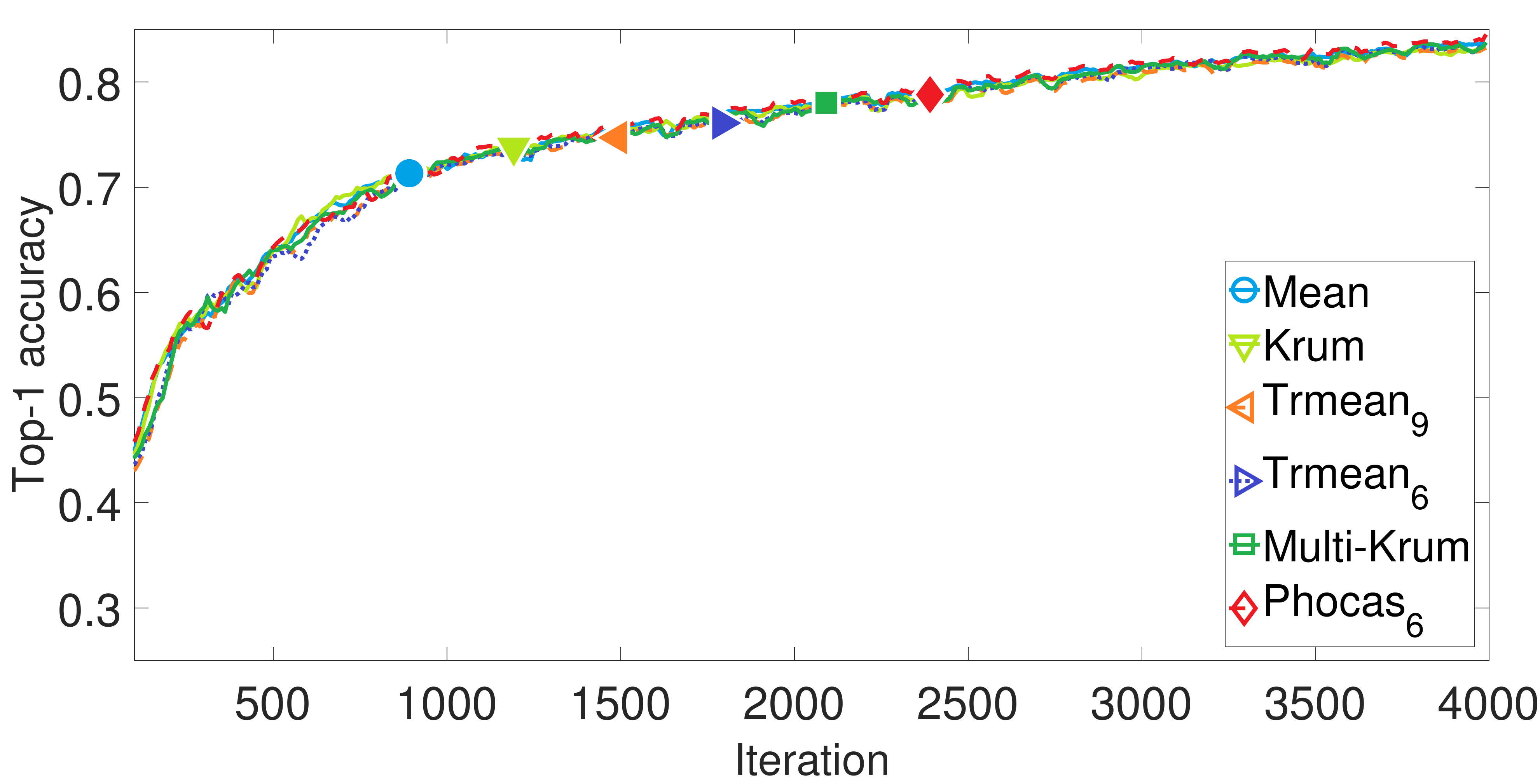}
\caption{Top-3 Accuracy of CNN VS. \# rounds evaluated on CIFAR10 without Byzantine failures}
\label{fig:cifar10_nobyz_appendix}
\end{figure*}
\begin{figure*}[htb]
\centering
\includegraphics[width=0.90\textwidth]{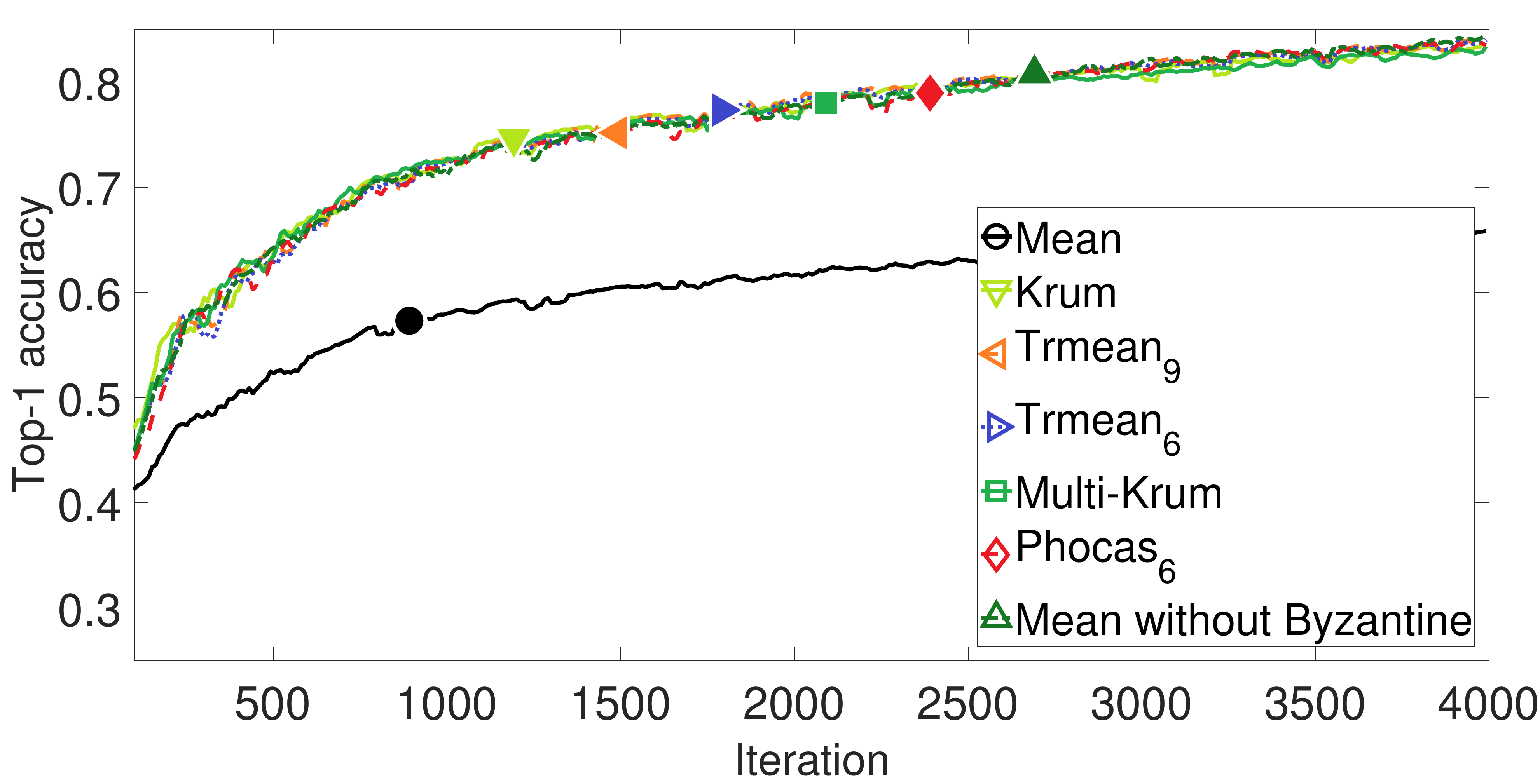}
\caption{Top-3 Accuracy of CNN VS. \# rounds evaluated on CIFAR10 with Gaussian attack}
\label{fig:cifar10_gaussian_appendix}
\end{figure*}
\begin{figure*}[htb]
\centering
\includegraphics[width=0.90\textwidth]{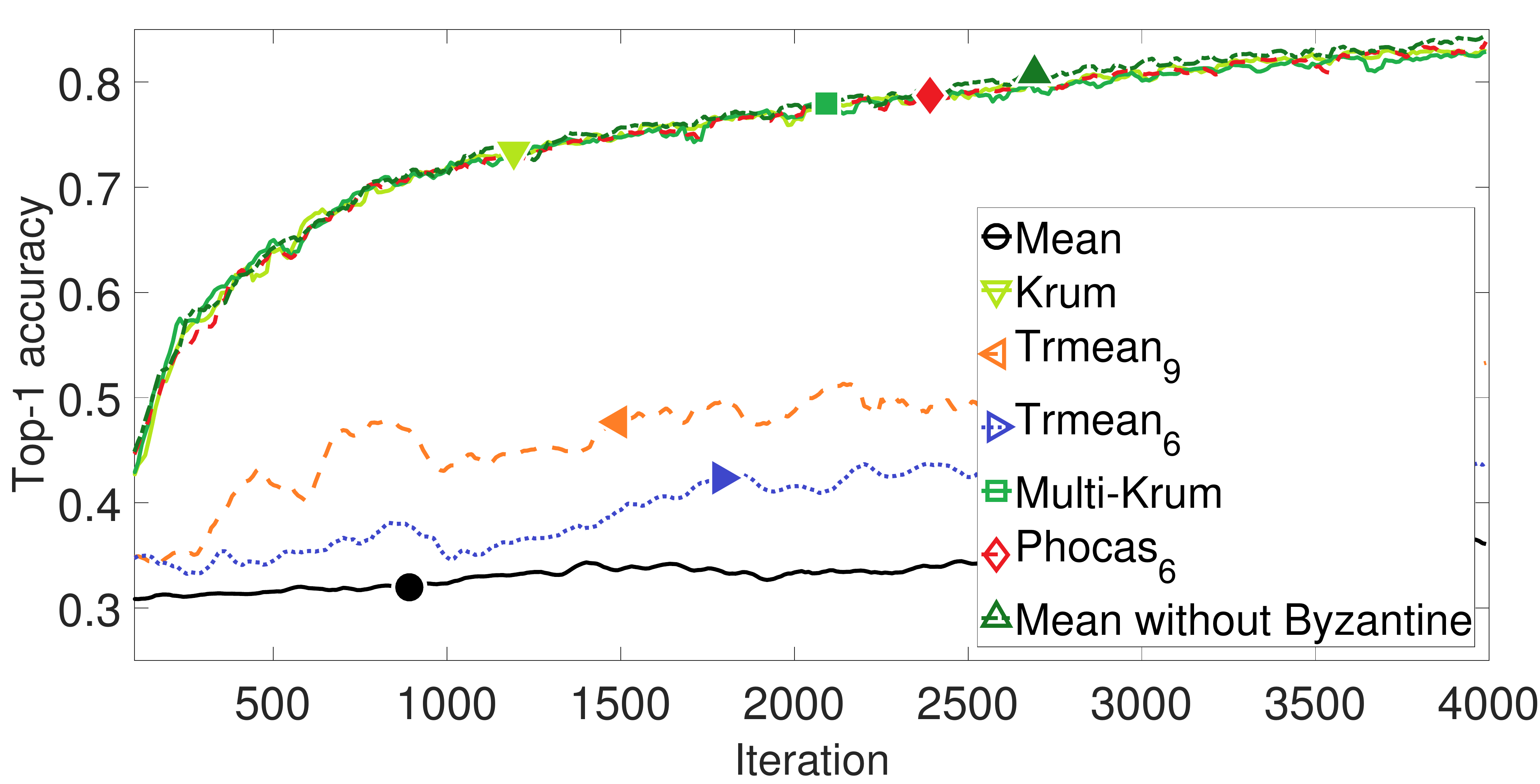}
\caption{Top-3 Accuracy of CNN VS. \# rounds evaluated on CIFAR10 with omniscient attack}
\label{fig:cifar10_omniscient_appendix}
\end{figure*}
\begin{figure*}[htb]
\centering
\includegraphics[width=0.90\textwidth]{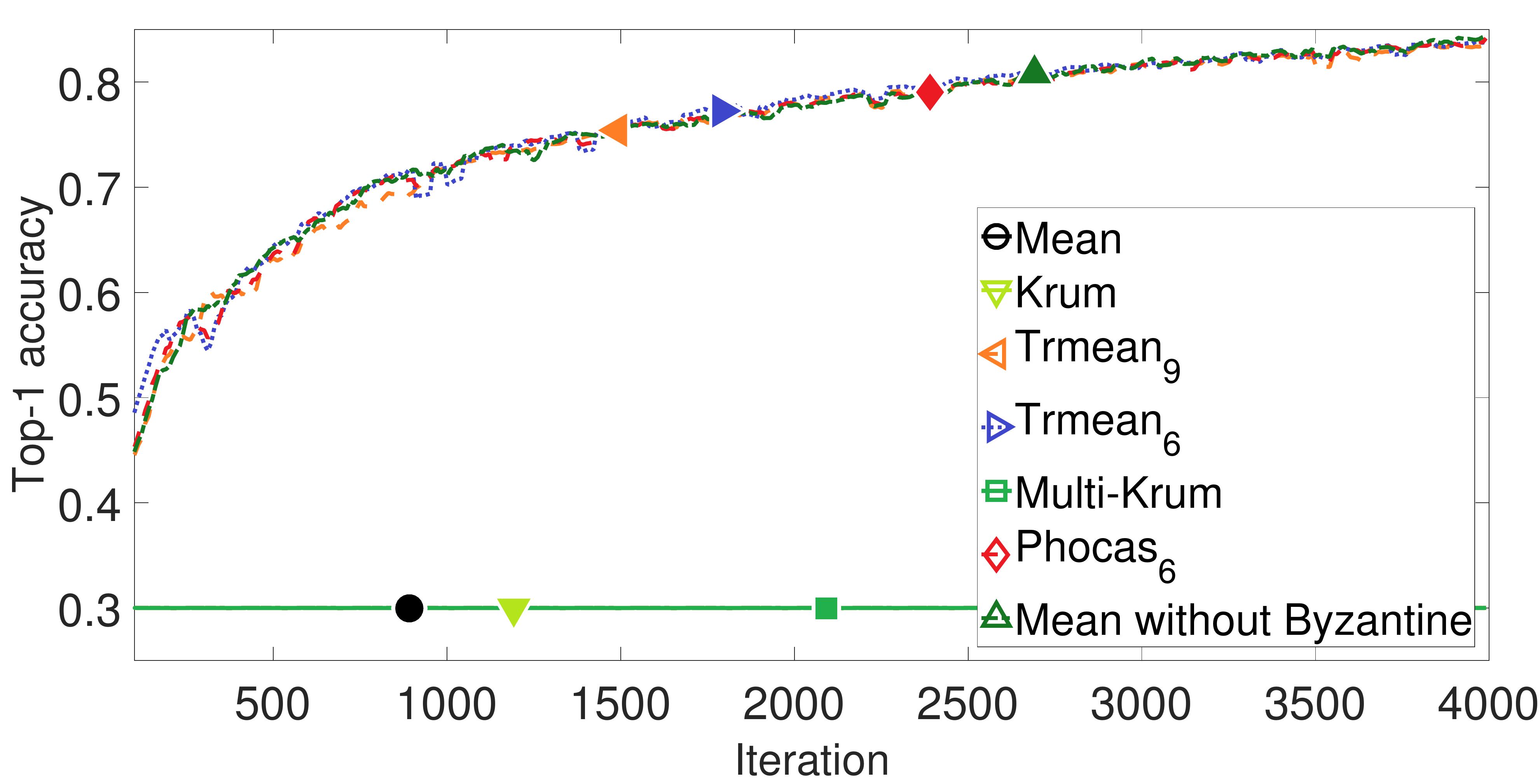}
\caption{Top-3 Accuracy of CNN VS. \# rounds evaluated on CIFAR10 with bit-flip attack}
\label{fig:cifar10_bitflip_appendix}
\end{figure*}
\begin{figure*}[htb]
\centering
\includegraphics[width=0.90\textwidth]{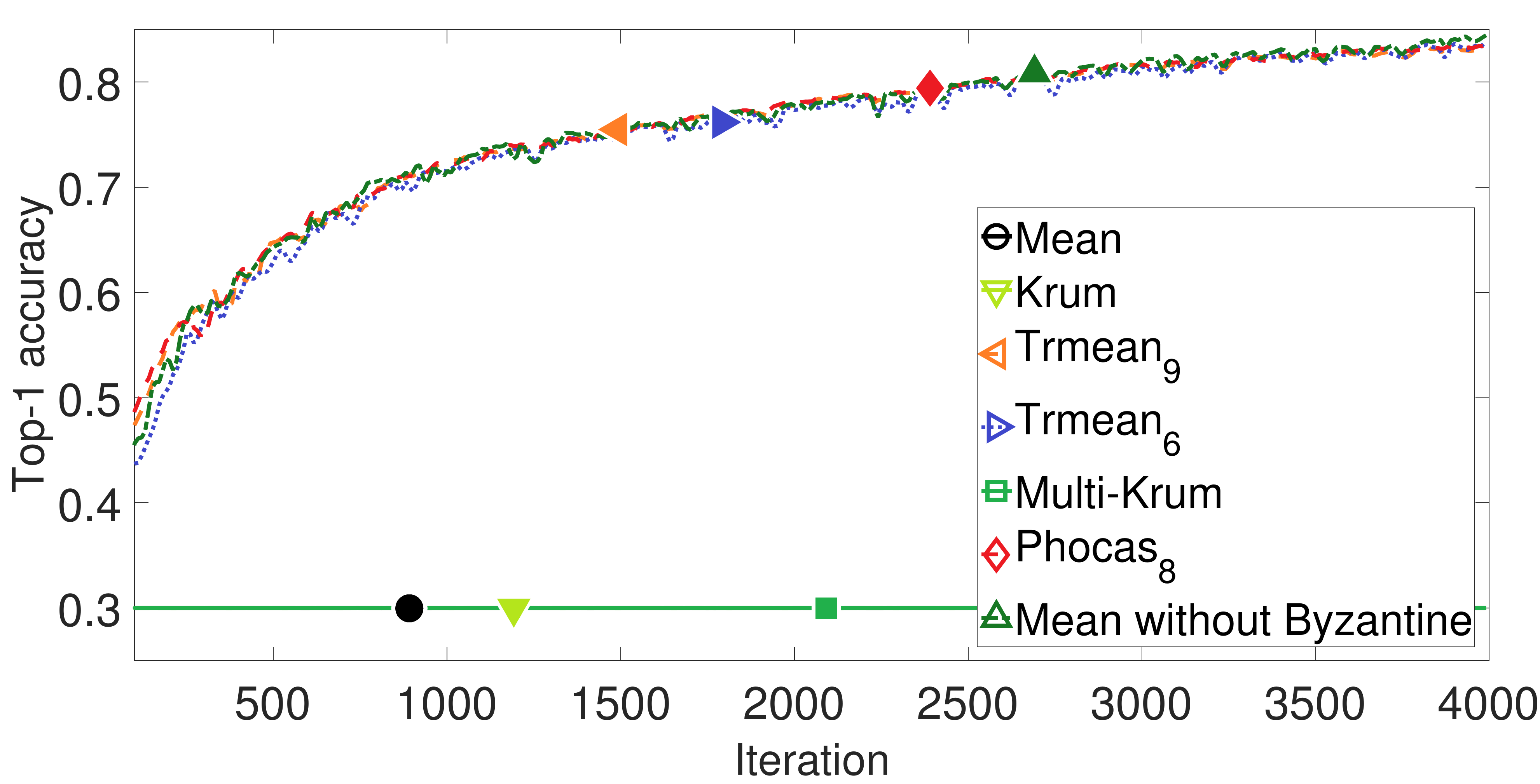}
\caption{Top-3 Accuracy of CNN VS. \# rounds evaluated on CIFAR10 with gambler attack.}
\label{fig:cifar10_multiserver_appendix}
\vspace{-0.5cm}
\end{figure*}

\end{document}